\documentclass[a4paper, USenglish, cleveref, autoref, thm-restate]{lipics-v2021}
\usepackage{import}
\usepackage{xcolor}
\usepackage{tikz}
\usetikzlibrary{automata} 
\usetikzlibrary{arrows.meta,automata,positioning}
\nolinenumbers
\title{Hamming distance between finite transducers}

\author{Luc Dartois}{Université Marie et Louis Pasteur, CNRS, institut FEMTO-ST, F- 25000 Besançon, France}{luc.dartois@femto-st.fr}{https://orcid.org/0000-0001-9974-1922}{}
\author{Pierre-Cyrille Héam}{Université Marie et Louis Pasteur, CNRS, institut FEMTO-ST, F- 25000 Besançon, France }{pierre-cyrille.heam@femto-st.fr}{https://orcid.org/0000-0002-1125-1767}{}
\author{Ismaël Jecker}{Université Marie et Louis Pasteur, CNRS, institut FEMTO-ST, F- 25000 Besançon, France }{ismael.jecker@femto-st.fr}{https://orcid.org/0000-0002-6527-4470}{This research was partially funded by the Agence Nationale de la Recherche (ANR) grant ANR-25-CE48-2447 FAVOR}
\author{Silvio Vescovo}{Université Marie et Louis Pasteur, CNRS, institut FEMTO-ST, F- 25000 Besançon, France }{silvio.vescovo@femto-st.fr}{https://orcid.org/0009-0002-4686-1677}{}

\authorrunning{L. Dartois, P.-C. Héam, I. Jecker, and S. Vescovo}


\hideLIPIcs  
\Copyright{Luc Dartois, Pierre-Cyrille Héam, Ismaël Jecker, Silvio Vescovo}
\EventEditors{John Q. Open and Joan R. Access}
\EventNoEds{2}
\EventLongTitle{Mathematical Foundations of Computer Science}
\EventShortTitle{MFCS}
\EventAcronym{MFCS}
\EventYear{2016}
\EventDate{December 24--27, 2016}
\EventLocation{Little Whinging, United Kingdom}
\EventLogo{}
\SeriesVolume{42}
\ArticleNo{23}

\usepackage{amsmath}
\usepackage{amsthm}

\usepackage{apxproof}
\newtheoremrep{prop}[proposition]{Proposition}
\newtheoremrep{lem}[lemma]{Lemma}
\newtheoremrep{theo}[theorem]{Theorem}
\newtheoremrep{coro}[corollary]{Corollary}

\renewcommand{\leq}{\leqslant}
\renewcommand{\geq}{\geqslant}

\newcommand{\Cal}[1]{\mathcal{#1}}

\newcommand{\Rm}[1]{\mathrm{#1}}

\newcommand{\true}{\textsc{True}}
\newcommand{\false}{\textsc{False}}
\newcommand{\dev}{\textsf{dev}}

\newcommand{\pspace}{\textsc{PSpace}}
\newcommand{\logspace}{\textsc{LogSpace}}

\DeclareMathOperator*{\dom}{dom}

\DeclareMathOperator*{\shift}{s}

\DeclareMathOperator*{\smax}{s_{max}}
\DeclareMathOperator*{\lmax}{\ell_{max}}

\DeclareMathOperator*{\inn}{in}
\DeclareMathOperator*{\out}{out}

\usepackage{todonotes}
  \newcounter{todocounter}

\usepackage{colortbl}
\usepackage{adjustbox}
\usepackage[most]{tcolorbox}

\newenvironment{boitebleue}{\begin{tcolorbox}[arc=1mm,colback=blue!10,boxsep=1pt,left=2pt,right=2pt,top=2pt,bottom=0pt]}{\end{tcolorbox}}

\newcommand{\PB}[3]{
\begin{boitebleue}
\noindent {\bf \textsf{#1}}\\
\noindent {\bf \textsf{\makebox[1.3cm][l]{Input:}}} #2\\
\noindent {\bf \textsf{\makebox[1.3cm][l]{Output:}}} #3
\end{boitebleue}
}

\begin{document}

\maketitle

\begin{abstract}
	We study bounded deviation of non-deterministic finite transducers under the Hamming distance:
	the bounded comparison problem asks, given two transducers and \(k \in \mathbb{N}\),
	whether for every input the two transducers produce words at Hamming distance at most \(k\).
	This problem is known to be decidable in polynomial time when \(k\) is fixed, and in co-NP otherwise.
	
	We show that the problem is NL-complete when \(k\) is fixed, co-NP-complete when \(k\) is given in binary, and it is DP-complete to decide if the distance is exactly \(k\). 
	We also prove that if the two transducers have bounded comparison, then the maximal distance is at most quadratic in the size of both transducers, and that this bound is asymptotically tight.
	
	We prove the results on deviations problems, 
	which asks similar questions on the distance of the pairs of input and output of a single transducer, and show that these two families of problems are logspace many-one equivalent.
	
\end{abstract}

\newpage

\section{Introduction} 
\label{sec:introduction}

Non-deterministic finite-state transducers (NFT) are a fundamental model for describing transformations between words.
Studied since the early days of computer science, these machines, initially known as
\emph{generalized sequential machines}~\cite{Raney1958SequentialF,Ginsburg1968},
are obtained by equipping transitions of finite-state automata with output words.
Whereas an automaton \(\Cal{A}\) recognizes a \emph{languages} \(L_{\Cal{A}}\),
a transducer \(\Cal{T}\) recognizes a \emph{binary relation} \(R_{\Cal{T}}\) between input and output words,
called a \emph{rational relation}.
We refer to~\cite{MuschollP19,FiliotR16} and the references therein for a comprehensive overview of this model.
As is standard in automata theory,
classical decision problems on transducers are inherently Boolean
(e.g., equivalence, or determinisability).
To move beyond this qualitative setting towards quantitative questions,
we need to determine a meaningful notion of \emph{distance between transducers}.
A key requirement is that such a notion should correspond to algorithmically tractable decision problems
to support effective analysis.
Lifting distances from words 
to transducers yields a natural candidate satisfying these requirements.

\smallskip
A common way of comparing two words \(u\) and \(v\) is through their \emph{edit distance} \(d(u,v)\),
defined as the minimum number of elementary operations, called \emph{edits},
required to transform \(u\) into \(v\).
Allowing different edits induce different distances.
The most well-known, called \emph{Levenshtein distance}~\cite{Levenshtein66},
allows insertions, deletions and substitutions,
and numerous variants arise by restricting operations
or assigning them weights.
In this work, we focus on the \emph{Hamming distance}~\cite{Hamming50},
where the only operation allowed is the substitution of a letter by another.
Beyond words, edit distances have been extended to richer structures.
The distance between two languages is typically defined as the minimal~\cite{HanKS12} or average~\cite{Mohri03}  distance
between pairs of words drawn from each language.
More recently, edit distances have been lifted to rational relations~\cite{AiswaryaMS24,FiliotJMS25}.
In this setting, the perspective shifts: rather than witnessing proximity via the existence of a close pair,
uniform closeness is required among all pairs of outputs associated with the same input.
Formally, given two relations \(R_1\) and \(R_2\), their distance is defined as \(\infty\)
if their domains are distinct, and otherwise we let
\[
d(R_1,R_2) = \sup{} \{d(v_1,v_2) \mid (u,v_1) \in R_1 \text{ and } (u,v_2) \in R_2
\text{ for some } u \in \Sigma^*\} \in \mathbb{N}\cup \{\infty\}.
\]
This worst-case viewpoint is well-suited to quantitative verification,
as it can capture guarantees on the deviation from an ideal behavior.
This distance between relations leads naturally to three fundamental decision problems:
determining whether the distance is finite,
whether it is bounded by a given threshold,
and whether it is exactly equal to a given value.

\PB{Bounded Comparison Problem}
{Two NFTs \(\Cal{T}_1, \Cal{T}_2\) with \(\dom(R_{\Cal{T}_1}) = \dom(R_{\Cal{T}_2})\).}
{\true{} if \(d(R_{\Cal{T}_1}, R_{\Cal{T}_2}) < \infty\), \false{} otherwise.}

\PB{Threshold-bounded Comparison Problem}
{Two NFTs \(\Cal{T}_1, \Cal{T}_2\) with \(\dom(R_{\Cal{T}_1}) = \dom(R_{\Cal{T}_2})\),
	and \(k \in \mathbb{N}\) encoded in binary.}
{\true{} if  \(d(R_{\Cal{T}_1}, R_{\Cal{T}_2}) \leq k\), \false{} otherwise.}

\PB{Exact Comparison Problem}
{Two NFTs \(\Cal{T}_1, \Cal{T}_2\) with \(\dom(R_{\Cal{T}_1}) = \dom(R_{\Cal{T}_2})\),
	and \(k \in \mathbb{N}\) encoded in binary.}
{\true{} if \(d(R_{\Cal{T}_1}, R_{\Cal{T}_2}) = k\), \false{} otherwise.}

\noindent
Restricting the inputs to transducers with identical domains is natural in this setting.
Indeed, if $\dom(R_{\Cal{T}_1}) \neq \dom(R_{\Cal{T}_2})$,
then the distance $d(R_{\Cal{T}_1},R_{\Cal{T}_2})$ is infinite by definition,
making the comparison trivial.
Moreover, this restriction allows us to isolate the intrinsic complexity of the comparison problems.
Without it, one would first need to check whether the domains coincide,
which is \pspace-complete \cite{Stockmeyer1973} and would therefore dominate the overall complexity.


\subparagraph*{Contributions and organization of the paper.}
In this paper we approach these problems from a different perspective:
instead of comparing the outputs of two relations
we compare the input and output of a single relation.
Formally, given a binary relation \(R\),
the \emph{deviation} of \(R\)
is the maximal Hamming distance between input and output
over all pairs in \(R\):
\[
	\dev(R) = \sup{} \{d(u,v) | (u,v) \in R\} \in \mathbb{N}\cup \{\infty\}.
\]
This notion naturally gives rise to the following decision problems.
\PB{Bounded Deviation Problem}{An NFT \(\Cal{T}\).}
{\true{} if \(\dev(R_\Cal{T}) < \infty\), \false{} otherwise.}

\PB{Threshold-bounded Deviation Problem}
{An NFT \(\Cal{T}\) and \(k \in \mathbb{N}\) encoded in binary.}
{\true{} if \(\dev(R_\Cal{T}) \leq k\), \false{} otherwise.}

\PB{Exact Deviation Problem}{An NFT \(\Cal{T}\) and \(k \in \mathbb{N}\) encoded in binary.}
{\true{} if \(\dev(R_\Cal{T}) = k\), \false{} otherwise.}

\noindent
We first show that these problems are equivalent to their comparison counterparts.

\begin{theorem}\label{thm:equivalent_problems}
	The Bounded, Threshold-bounded and Exact Deviation Problems
	are logspace many-one equivalent to the corresponding Comparison Problems.
\end{theorem}

\noindent
Theorem~\ref{thm:equivalent_problems} is proved in Section~\ref{sec:application_to_the_k_similarity_problem},
where we show how to construct,
from two transducers \(\Cal{T}_1,\Cal{T}_2\), a single transducer \(\Cal{T}\)
such that \(\dev(R_{\Cal{T}}) = d(R_{\Cal{T}_1},R_{\Cal{T}_2})\) (Proposition~\ref{prop:ComparisonToDeviation}),
and conversely how to construct, from a transducer \(\Cal{T}\),
a pair of transducers \(\Cal{T}_1,\Cal{T}_2\)
such that \(\dev(R_{\Cal{T}}) = d(R_{\Cal{T}_1},R_{\Cal{T}_2})\) (Proposition~\ref{prop:deviationToComparison}).
We then establish tight complexity bounds for all problems.

\begin{theorem}\label{thm-BoundedNFT}
	The Bounded Deviation Problem for the Hamming distance is NL-complete.
\end{theorem}

\begin{theorem}\label{thm-KboundedNFT}
	The Threshold-bounded Deviation Problem for the Hamming distance is:
	\begin{itemize}
		\item co-NP-complete when \(k\) is part of the input;
		\item NL-complete for every fixed \(k \geq 1\), when \(k\) is an outside constant.  
	\end{itemize}
\end{theorem}

\noindent
The DP complexity class, for \emph{difference polynomial time} is first defined in \cite[S.~2]{Papadimitriou1984} as the class of the problems that are expressed as the difference between two NP problems.
Alternatively, it consists of problems that can be defined as the intersection of an NP problem and a co-NP problem. Note that DP contains both NP and co-NP.

\begin{theorem}\label{thm-Kexact}
	The Exact Deviation Problem for the Hamming distance is DP-complete.
\end{theorem}

\noindent
Note that, by Theorem~\ref{thm:equivalent_problems},
these results immediately extend to the Comparison Problems.
Theorems~\ref{thm-BoundedNFT}-\ref{thm-Kexact} are proved over two sections:
in Section~\ref{sec:lowerbounds}, we prove the matching hardness results
by reductions from canonical complete problems
(Propositions~\ref{prop:bounded-hardness}, \ref{prop:fixed-hardness}, \ref{prop:threshold-bounded-hardness} and \ref{prop:exact-hardness}).
In Section~\ref{sec:upperBounds},
we define algorithms establishing membership of the Bounded Problems in the corresponding complexity class
(Propositions~\ref{prop:BoundedNFTHigh} and \ref{prop:KboundedNFTHigh}).
Remark that the Exact Deviation Problem is DP as it is the intersection of the Threshold Bounded Problem which is co-NP, and its complement~\cite{SainaThesis}.
We also show a quadratic bound on the size of the transducer for the deviation of bounded NFT.

\begin{theorem}\label{thm-tightBound}
	For every NFT \(\Cal{T}\), if \(\dev(R_{\Cal{T}}) < \infty\) then \(\dev(R_{\Cal{T}}) = O(|\Cal{T}|^2)\).
	Moreover, there exists a family \((\Cal{T}_{n})_{n \in \mathbb{N}}\)
	such that each \(\Cal{T}_n\) has \(2n\) states, \(3n-1\) atomic transitions, and satisfies
	\(\dev(R_{\Cal{T}_i}) = \frac{n^2+n}{2}\).
\end{theorem}

The upper bound is shown in Proposition~\ref{prop:BoundedNFTHigh}.
The lower bound is a consequence of Lemma~\ref{lem:asymptoticFamily}.


\subparagraph*{Related work.}
The Comparison Problems for both Hamming and Levenshtein distance were already studied in~\cite{AiswaryaMS24},
with additional details in the long version~\cite{AiswaryaMS24LongVersion}
and in the PhD thesis of S.~Sunny~\cite{SainaThesis}.
In particular, it is shown that for the Hamming distance:
\begin{enumerate}
	\item the Bounded Comparison Problem is decidable in polynomial time~\cite[Theorem 4.10]{AiswaryaMS24},
	\item the Threshold-bounded Comparison Problem is in co-NP~\cite[Theorem 4.13]{AiswaryaMS24LongVersion},
	\item the Exact Comparison Problem is in DP~\cite[Theorem 6.19]{SainaThesis}.
\end{enumerate}
These results correspond to the upper bounds we revisit in Section~\ref{sec:upperBounds}.
Our contributions strengthen them as follows.
For the Bounded problem,
while we use a similar proof structure and characterization,
we provide alternative proofs,
and we show that the problem is actually in NL (Proposition~\ref{prop:BoundedNFTHigh}).
Moreover, we establish a tight quadratic bound on the distance when it is finite (Theorem~\ref{thm-tightBound}).
Regarding this bound, although no explicit statement appears in~\cite{AiswaryaMS24LongVersion},
a polynomial bound can be extracted from the proofs therein.
We make this bound explicit and show that it is optimal.
For the Threshold-bounded Problem,
we identify the true source of intractability:
the parameter \(k\), rather than the transducers themselves.
More precisely, we show that the Threshold-bounded Comparison problem drops from co-NP to NL when \(k\) is fixed (Proposition~\ref{prop:KboundedNFTHigh}).
Note that these similarities with~\cite{AiswaryaMS24} concern only Section~\ref{sec:upperBounds},
whereas Section~\ref{sec:lowerbounds} contains entirely new results.

\smallskip
\noindent
Other extensions of edit distance to rational relations have been considered.

The notion studied in this paper is inherently \emph{universal}:
for every input, \emph{all} corresponding outputs must be close.
In contrast, the notion of \emph{almost reflexivity}, introduced in~\cite{ChoffrutP02},
adds an existential quantifier:
for every input, there must exist \emph{at least one} corresponding output that is close.
This notion is less well-behaved algorithmically,
as most related decision problems are undecidable over the class of rational relations.

The \emph{robustness} introduced in~\cite{SamantaDC13,HenzingerOS14}
proposes another way of comparing transducers based on edit distances.
Rather than comparing outputs corresponding to a fixed input, like the Comparison Problems,
or comparing input and output, like the Deviation Problems,
it enforces a Lipschitz-like condition,
requiring that close inputs yield close outputs.
Again, this notion is algorithmically less well-behaved and leads to undecidability in general:
restrictions to subclasses of rational relations are required to recover decidability.


\section{Preliminaries} 
\label{sec:preliminaries}

\subparagraph*{Words and Hamming distance.}
An \emph{alphabet} is a finite set of symbols called \emph{letters}.
A \emph{word} over a given alphabet \(\Sigma\) is a finite sequence of letters of \(\Sigma\).
The \emph{empty word}, denoted \(\varepsilon\), is the empty sequence.
The set of all words over an alphabet \(\Sigma\) is the free monoid \(\Sigma^{*}\).
Given a word \(u\), we denote by \(|u|\) the length of the sequence of letters of \(u\), i.e.
the number of letters in \(u\), and for each \(i\) in \(1 \leq i \leq |u|\) we denote by \(u_{i}\) the i-th letter of \(u\).
For two words \(u = u_{1}u_{2}\cdots u_{n}\), \(v = v_{1}v_{2}\cdots v_{m}\) we denote \(uv\) the concatenation of \(u\) and \(v\): \(uv = u_{1}\cdots u_{n}v_{1}\cdots v_{m}\).
Note that the concatenation operation is associative and \(\varepsilon\) is the neutral element for this operation.
For a given word \(u \in \Sigma^{*}\), we say that a word \(x \in \Sigma^{*}\) is a \emph{factor} of \(u\) if there exist two words \(v\) and \(w\) of \(\Sigma^{*}\) such that \(u = vxw\).
A \emph{prefix} (resp. \emph{suffix}) \(x \in \Sigma^{*}\) of \(u \in \Sigma^{*}\) is a factor  of \(u\) where \(v\) (resp. \(w\)) is the empty word.
A word \(v_1\cdots v_n\) is a \emph{sub-word} of a word \(u\) if there exist some words \(w_0,\ldots,w_n\) such that \(u=w_0v_1w_1\cdots v_nw_n\).
We say that two words \(u\) and \(v\) are \emph{conjugate} if there exist a prefix \(w\) of \(u\) and a prefix \(z\) of \(v\) such that \(u = wz\) and \(v = zw\).
Furthermore, we say that \(u\) is \emph{conjugate to \(v\) by \(n \in \mathbb{Z}\)} if \(u\) and \(v\) are in fact conjugate and for all \(i,j \in [1, |u|]\) such that \(j-i \equiv n \bmod |u|\), \(u_i = v_j\)\footnote{Note that it is symmetric, if \(u\) is conjugate to \(v\) by \(0\leq n<|u|\), then \(v\) is conjugate to \(u\) by \(|u| - n\)}.
In this article, we are interested in deciding the similarity between machines and their word outputs, using the \emph{Hamming distance} \cite[S.~3.6]{Hamming1986}. To this end, we define the function \(d\) from \((u,v) \in \Sigma^* \times \Sigma^*\) to \( d(u,v) \in \mathbb{N} \cup \{+\infty\}\) where:
	\[
		d(u,v) = 
		\begin{cases}
			+\infty, &\text{ if } |u| \neq |v|\\
			|\{i \in [1,|u|] \mid u_i \neq v_i\}|
			&\text{ if } |u| = |v|\\
		\end{cases}
	\]
	Note that the function \(d\) is equal to the Hamming distance when \(|u| = |v|\).

\subparagraph*{Non-deterministic Finite-state Transducers.}
A \emph{non-deterministic finite-state transducer} (NFT) \(\Cal{T}\) is an extension of a finite automaton, i.e. a quintuplet \((Q,\Sigma, Q_{i}, Q_{f}, \Delta)\) where \(Q\) is a finite set of states, \(\Sigma\) is both the input and output alphabet, \(Q_i\subseteq Q\) and \(Q_f\subseteq Q\) are respectively the sets of initial and final states. The set of transitions \(\Delta\) is a finite subset of \(Q \times \Sigma^* \times \Sigma^* \times Q\). Note that generally the input and output alphabets are defined as different. However, this definition is without loss of generality as one can always consider \(\Sigma\) to be the union of the input and output alphabets.

A \emph{run} \(\rho\) of \(\Cal{T}\) is a word \(\delta_{1}\delta_{2}\cdots\delta_{n}\) in \(\Delta^{*}\) such that for all \(1 \leq i < n\), \(\pi_4(\delta_i) = \pi_1(\delta_{i+1})\), \emph{where \(\pi_{j}(\delta)\) is the projection on the j-th component of \(\delta\)}.
If \(n \geq 1\), we say that \(\rho\) is a run from a state \(p = \pi_1(\delta_1)\) to a state \(q = \pi_4(\delta_n)\) over the pair of words \((u,v)\), where \(u = \pi_2(\delta_1)\pi_2(\delta_2)\cdots\pi_2(\delta_n)\) and \(v = \pi_3(\delta_1)\pi_3(\delta_2)\cdots\pi_3(\delta_n)\).
The length of a run \(\rho\), denoted by \(|\rho|\), is the number of transitions of \(\rho\).
A run  \(\rho\) is said to be \emph{empty} if and only if \(|\rho| = 0\), in this case, the run \(\rho\) is a run from any state \(q \in Q\) to itself over the pair of words \((\varepsilon, \varepsilon)\).  
A run \(\rho\) from \(p\) to \(q\) is \emph{initial} if \(p \in Q_i\) and \emph{final} if \(q \in Q_f\).
A run that is both initial and final is an \emph{accepting} run of \(\Cal{T}\).

A transducer \(\Cal{T}\) defines a relation, denoted \(R_\Cal{T}\), as a subset of \(\Sigma^* \times \Sigma^*\) where \((u,v) \in R_\Cal{T}\) if and only if there exists an accepting run \(\rho\) over \((u,v)\).
The domain \(\dom(R_{\Cal{T}})\) of \(R_{\Cal{T}}\) is a subset of \(\Sigma^*\) defined as \(\{u \mid (u,v) \in R_{\Cal{T}}\}\).
An NFT \(\Cal{T}\) is \emph{length-preserving} if for all \((u,v)\) in \(R_{\Cal{T}}\), \(|u| = |v|\).
Given an integer  \(k\geq 0\), an NFT \(\Cal{T}\) is said to be \(k\)\emph{-bounded} if \(d(u,v) \leq k\) for all \((u,v) \in R_{\Cal{T}}\). It is \emph{bounded} if it is \(k\)-bounded for some \(k\geq 0\).
Two NFTs \(\Cal{T}\) and \(\Cal{T'}\) are said to be \emph{equivalent} if \(R_{\Cal{T}} = R_{\Cal{T'}}\).
Throughout this article, we assume that all NFTs are trimmed.
Note that automata and hence transducers
can be made trim in \logspace~using repeated calls
to the NL-complete Directed Graph Reachability Problem \cite{Papadimitriou94}.

In this article, we consider two metrics for the size of a transducer. When possible we only consider the number of states, denoted \(|Q|\). However, as we consider non-deterministic transducers that can read and output finite but arbitrarily long words, we sometime refer to the size of \(|\Cal{T}|\), by which we mean the total size needed to represent \(\Cal{T}\) on the tape of a Turing Machine. In particular, it takes into account \(Q\) but also all transitions, written as quadruplets where the indexes of the states are written in binary and the input and output words are written as such.

Finally, we also introduce two functions that will be used for some proofs. Given an NFT \(\Cal{T}\), and a run \(\rho\) of \(\Cal{T}\) over some \((u,v)\):
\begin{itemize}
	\item We define the function \(\inn_\rho\) from \([1,|u|]\) to \([1, |\rho|]\) associating an index \(i\) in \(u\) to the index of the transition reading the i-th letter of the input \(u_i\) in \(\rho\).
	\item Similarly, we define the function \(\out_\rho\) from \([1,|v|]\) to \([1, |\rho|]\) which associates an index \(j\) of \(v\) to the index of the transition writing the j-th letter of the output \(v_j\) in \(\rho\).
\end{itemize}
When the run $\rho$ is clear from context, we simply write \(\inn(i)\) and \(\out(j)\).

\subparagraph*{Shift of a run.} 
\label{sub:misc}
In this contribution, the \emph{shift} of a run is the difference in length between what is read
and what is produced.
Formally, the shift of a transition \(\delta = (p,u,v,q)\) of an NFT is denoted by \(\shift(\delta)\) and is equal to \(|u| - |v|\).
The \emph{maximum transition shift} of an NFT \(\Cal{T}\) is denoted by \(\smax(\Cal{T})\), and it is equal to \( \max\limits_{\delta \in \Delta}(|\shift(\delta)|)\).
The notion of shift extends to runs as follows: given a run \(\rho\) of an NFT over \((u,v)\), the \emph{shift of the run}, also denoted by \(\shift(\rho)\), is equal to \(|u| - |v|\).
\begin{remark}\label{remark:shift_def}
    Given a non-empty run \(\rho = \delta_1\delta_2\cdots\delta_n\) of an NFT \(\Cal{T}\) we have \(\shift(\rho) = \sum\limits_{i = 1}^{n}\shift(\delta_i)\).
\end{remark}
In parallel, the \emph{length} of a transition \(\delta = (p,u,v,q)\) of an NFT is denoted \(|\delta|\) and is equal to \(|u|+|v|\).
The \emph{maximum transition length} of an NFT \(\Cal{T}\) is denoted \(\lmax(\Cal{T})\) and is equal to
\(
	\max\limits_{\delta \in \Delta}(|\delta|).
\)

\begin{example}\label{ex-quadratic}
	We give an example of an NFT \(\Cal{T}_4\) in~\cref{fig:quadratic-mismatches-example}. It realizes the relation:
	\[
		R_{\Cal{T}_4} =
		\{
			(1^{n_1+1}0^{n_2+1}1^{n_3+1}0^{n_4+1},
			1^{n_1}0^{n_2}1^{n_3}0^{n_4}1^4)
			\mid n_1,n_2,n_3,n_4 \in \mathbb{N}
		\}.
	\]
	The pair \((1001110000,0110001111)\) is in \(R_{\Cal{T}_4}\)
	and produces \(\sum\limits_{i=1}^4i=10\) mismatches.
	We now generalize this construction.
\end{example}
	
\begin{lemma}\label{lem:asymptoticFamily}
	There exists a family of NFTs \((\Cal{T}_n)_{n \geq 2}\)
	such that for every \(n \geq 2\) the NFT \(\Cal{T}_n\) has \(2n\) states
	and satisfies \(\dev(R_{\Cal{T}}) \geq \frac{n^2+n}{2}\).
	Moreover, the NFT \(\Cal{T}_n\) has \(3n-1\) transitions and \(\smax(\Cal{T}_n)=1\).
\end{lemma}

\begin{proof}
	For every \(n \geq 2\), we define an NFT
	\(\Cal{T}_n = (\{p_1,\ldots,p_n\}\cup\{q_1,\ldots,q_n\},\{0,1\}, p_1, q_n, \Delta_n)\),
	where \(\Delta_n\) is defined as follows:
	\[
	\begin{array}{lll}
	\Delta_n = & \{
	\{(p_i,i\bmod 2,\varepsilon,p_{i+1}),(p_i,i\bmod 2,i\bmod 2,p_{i})\mid 1\leq i< n\} \\
	&\cup \{(p_n,(n+1\bmod 2),n\bmod 2,q_1)\}\\
	& \cup \ \{(q_i,\varepsilon,n\bmod 2,q_{i+1})\mid 1\leq i<n\}.
	\end{array}
	\]
	The transducer \(\Cal{T}_n\) has \(2n\) states and \(3n-1\) transitions. The maximal shift on any given transition is \(1\), and it recognizes the relation \(R_{\Cal{T}_n}\) equal to
	\[
	\{
	(1^{k_1+1}0^{k_2+1} \cdots (n\bmod 2)^{k_{n}+1},
	1^{k_1}0^{k_2}\cdots (n\bmod 2)^{k_{n-1}}(n+1\bmod 2)^n)
	\mid k_1,\ldots,k_n \in \mathbb{N}
	\}.
	\]
	The claimed bound \(\dev(R_{\Cal{T}_n}) \geq \frac{n^2+n}{2}\) is witnessed by the pair
	\[
	(10^{2}1^{3}\cdots (n-1\bmod 2)^{n-1}(n\bmod 2)^{n},
	01^20^3\cdots (n\bmod 2)^{n-1}(n+1\bmod 2)^n) \in R_{\Cal{T}_n}
	\]
	which occurs by setting \(k_i = i-1\)
	for every \(1 \leq i \leq n-1\).
\end{proof}

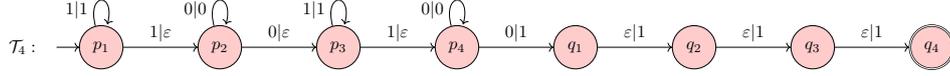
\begin{figure}
	\centering
	\resizebox{0.9\textwidth}{!}{

\begin{tikzpicture}[node distance = 1.5cm, auto,initial text=]
	\node[state,initial, fill=red!20](q_1){\(p_1\)};
	\node[state, fill=red!20](q_2)[right=of q_1]{\(p_2\)};
	\node[state, fill=red!20](q_3)[right=of q_2]{\(p_3\)};
	\node[state, fill=red!20](q_4)[right=of q_3]{\(p_4\)};
	\node[state, fill=red!20](q_5)[right=of q_4]{\(q_1\)};
	\node[state, fill=red!20](q_6)[right=of q_5]{\(q_2\)};
	\node[state, fill=red!20](q_7)[right=of q_6]{\(q_3\)};
	\node[state, fill=red!20, accepting](q_8)[right=of q_7]{\(q_4\)};
	\node[](t)[left=0.75cm of q_1]{\(\Cal{T}_4:\)};

	\path[->] (q_1) edge [loop above]	node [left, pos=0.2]	{\(1|1\)}			(q_1)
			  (q_1) edge				node [above]			{\(1|\varepsilon\)}	(q_2)
			  (q_2) edge [loop above]	node [left, pos=0.2]	{\(0|0\)}			(q_2)
			  (q_2) edge				node [above]			{\(0|\varepsilon\)} (q_3)
			  (q_3) edge [loop above]	node [left, pos=0.2]	{\(1|1\)}			(q_3)
			  (q_3) edge				node [above]			{\(1|\varepsilon\)}	(q_4)
			  (q_4) edge [loop above]	node [left, pos=0.2]	{\(0|0\)}			(q_4)
			  (q_4) edge 				node [above]			{\(0|1\)}			(q_5)
			  (q_5) edge 				node [above]			{\(\varepsilon|1\)}	(q_6)
			  (q_6) edge 				node [above]			{\(\varepsilon|1\)}	(q_7)
			  (q_7) edge 				node [above]			{\(\varepsilon|1\)}	(q_8);
\end{tikzpicture}}
	\caption{A bounded NFT with 8 states, 11 atomic transitions, creating at most 10 mismatches.}
	\label{fig:quadratic-mismatches-example}
\end{figure}

\subsection{Input-atomic transducer} 
\label{sub:atomic_transducer}

We introduce the notion of \emph{input-atomic transducer} as it will be used to simplify some proofs in the \cref{sec:application_to_the_k_similarity_problem}.
An input-atomic transducer is an NFT \(\Cal{T} = (Q,\Sigma, Q_{i}, Q_{f}, \Delta)\) where all transitions in \(\Delta\) read at most one letter.

\begin{lemma}\label{lemma:atomic_transducer_construction}
    For every NFT, an equivalent input-atomic transducer can be constructed in \logspace.
\end{lemma}

\begin{proof}
	From an NFT \(\Cal{T} = (Q,\Sigma, Q_{i}, Q_{f}, \Delta)\), we construct an input-atomic NFT \(\Cal{T'} = (Q',\Sigma, Q_{i}, Q_{f}, \Delta')\) such that \(\Cal{T'}\) and \(\Cal{T}\) are equivalent.
	We set \(Q' \subseteq Q \cup \Delta \times \lmax(\Cal{T})\) and \(\Delta'\) is equal to
	\[
		\bigcup\limits_{\delta \in \Delta} \Delta'_\delta,
	\]
	where for all transition \(\delta = (p,u,v,q) \in \Delta\), abusing the notation that \(p = (\delta,0)\), \(q = (\delta, |u|)\) we define:
	\[
		\begin{cases}
			\Delta'_{\delta} = \{\delta\}, &\text{ if }|u| \leq 1.\\
			\Delta'_{\delta} = \{((\delta,0),u_1,v,(\delta,1))\} \cup \{(\delta,i),u_{i+1},\varepsilon,(\delta,i+1)) | 1\leq i<|u|\}, &\text{ if }|u| > 1.\\
			
		\end{cases}
	\]
	By construction, for all transitions \(\delta = (p,u,v,q) \in \Delta\), there exists a run \(\rho = (p',u',v',q') \in \Delta'^*\) such that \(p'=p\), \(q'=q\), \(u'=u\), and \(v'=v\).
	Consequently, \(\Cal{T}\) and \(\Cal{T}'\) are equivalent.
	Note that in order to construct each transition \(\Cal{T}'\) we only need to store two indexes of states and read the transition letter by letter. Consequently, it is in \logspace. 
\end{proof}



\section{Equivalence of the Comparison and Deviation Problems} 
\label{sec:application_to_the_k_similarity_problem}

This section is devoted to the proof of~\cref{thm:equivalent_problems}, i.e. the two-way logspace many-one reductions between deviation and comparison problems.
More precisely, we show two reductions that preserve the distance and hence the exact bound. As such, the same reductions can be applied to show equivalence of the three pairs of problems.

\begin{proposition}\label{prop:ComparisonToDeviation}
	For all pairs of NFT \(\Cal{T}_1\) and \(\Cal{T}_2\) such that \(\dom(\Cal{T}_1) = \dom(\Cal{T}_2)\),
	there exists an NFT \(\Cal{T}\) such that \((u,v) \in R_{\Cal{T}}\) if and only if there exists an input word \(x\) satisfying \((x,u) \in R_{\Cal{T}_1}\) and \((x,v) \in R_{\Cal{T}_2}\).
	Furthermore, \(\Cal{T}\) is computable in \logspace.
\end{proposition}

\begin{proof}
To achieve the reduction we first transform both NFTs to reach a state where there are both input-atomic transducers, i.e. each transition either read a letter or \(\varepsilon\).
Next we add \((\varepsilon,\varepsilon)\)-transitions from each state to itself.
After those two transformations, we compute the composition of the two resulting NFTs.

	Let \(\Cal{T} = (Q,\Sigma, Q_{i}, Q_{f}, \Delta)\), and \(\Cal{S} = (P,\Sigma, P_{i}, P_{f}, \Theta)\) be two NFTs such that \(\dom(\Cal{T}) = \dom(\Cal{S})\).
	We transform \(\Cal{T}\) and \(\Cal{S}\) into \(\Cal{T}'\) and \(\Cal{S}'\) respectively, two input-atomic transducers such that \(R_{\Cal{T}'} = R_{\Cal{T}}\), and \(R_{\Cal{S}'} = R_{\Cal{S}}\). We also add \((\varepsilon,\varepsilon)\)-transitions to all states in both \(\Cal{S'}\) and \(\Cal{T}'\): \(\Theta' = \Theta' \cup \{(q,\varepsilon,\varepsilon,q) \mid q \in P'\}\), and \(\Delta' = \Delta \cup \{(q,\varepsilon,\varepsilon,q) \mid q \in Q'\}\).
	Thanks to~\cref{lemma:atomic_transducer_construction} these constructions are in \logspace.
	Finally, we construct the NFT \(\Cal{Z} = (O,\Sigma,O_i,O_f,\Lambda)\) such that \(O = Q' \times P'\), \(O_i = Q'_i \times P'_i\), \(O_f = Q'_f \times P'_f\), and \(\Lambda \subseteq O \times \Sigma^* \times \Sigma^* \times O\), with \(\Lambda = \{((q,p),u,v,(q',p')) \mid \exists x\ (q,x,u,q') \in \Delta'' \wedge (p,x,v,p') \in \Theta'\}\).
	By construction, we have that \((u,v) \in R_{\Cal{Z}}\) if and only if \((x,u) \in R_{\Cal{T}}\), and \((x,v) \in R_{\Cal{S}}\).
	
	Constructing the set of states can be done by storing the indexes of both states, hence it is in \logspace, and constructing the set of transitions can be done by storing the indexes of states and the value \(x\in \Sigma\cup\{\varepsilon\}\) and reading and writing both transitions.
	The composition of this transformation with the one from~\cref{lemma:atomic_transducer_construction} can be done in \logspace~in the following way:
	the algorithm simulates the second machine and the position of the reading head of the simulation in \logspace. Each time the simulation asks for the \(k\)-th bit of its input, the algorithm launches a \logspace~simulation of the first machine up to its \(k\)-th production.
	The complete construction is therefore in \logspace.
\end{proof}

\begin{proposition}\label{prop:deviationToComparison}
	For all NFT \(\Cal{T}\), there exists two NFTs \(\Cal{T}_1\) and \(\Cal{T}_2\) computable in \(O(1)\)-space such that \(\dom(\Cal{T}_1) = \dom(\Cal{T}_2)\),
	and \((u,v) \in R_{\Cal{T}}\) if and only if
	there exists an input word \(x\) satisfying \((x,u) \in R_{\Cal{T}_1}\), \((x,v) \in R_{\Cal{T}_2}\).

\end{proposition}

\begin{proof}
	Given an NFT \(\Cal{T} = (Q,\Sigma, Q_{i}, Q_{f}, \Delta)\),
	we let \(\Cal{T}_2 = \Cal{T}\), and
	we construct \(\Cal{T}_1 = (Q,\Sigma, Q_{i}, Q_{f}, \Delta')\),
	where \(\Delta' =\{(q,u,u,p) \mid (q,u,v,p) \in \Delta\}\).
	Then \(R_{\Cal{T}_1}\) is the identity function restricted to the domain of \(R_{\cal{T}}\), 
	hence \((u,v) \in R_{\Cal{T}}\) if and only if there exists \(x\), namely \(x = u\),
	such that \((x,u) \in R_{\Cal{T}_1}\) and \((x,v) \in R_{\Cal{T}_2}\).
	Constructing the NFTs \(\Cal{T}_i\) amounts to read and copy it twice, which can be done by a two-way transducer, and hence in constant space.
\end{proof}

\begin{figure}
	\centering
	\resizebox{0.9\textwidth}{!}{

\begin{tikzpicture}

\begin{scope}
\node[state,initial, initial text=,fill=blue!20] (1) at (0,0) {$1$};
\node[state,fill=blue!20] (2) at (2,0) {$2$};
\node[state,fill=blue!20] (3) at (2,2) {$3$};
\node[state,fill=blue!20] (4) at (4,0) {$4$};
\node[state,fill=blue!20] (5) at (6,0) {$5$};
\node[state,fill=blue!20,accepting] (6) at (6,2) {$6$};

\path[draw,thick,-latex] (1) edge node[above] {$a|\varepsilon$} (2);
\path[draw,thick,-latex] (2) edge node[above] {$\varepsilon|b$} (4);
\path[draw,thick,-latex] (4) edge node[above] {$\varepsilon|b$} (5);
\path[draw,thick,-latex] (5) edge node[right] {$b|\varepsilon$} (6);
\path[draw,thick,-latex] (2) edge[bend left] node[left] {$ab|\varepsilon$} (3);
\path[draw,thick,-latex] (3) edge[bend left] node[right] {$\varepsilon|ba$} (2);

\node(T1) at (4,1.5) {$T_1$};
\end{scope}
\begin{scope}[xshift=8.5cm]
\node[state,initial, initial text=,fill=blue!20] (1) at (0,0) {$1$};
\node[state,fill=blue!20] (2) at (2,0) {$2$};
\node[state,fill=blue!20] (3) at (2,2) {$3$};
\node[state,fill=blue!20] (4) at (4,0) {$4$};
\node[state,fill=blue!20] (5) at (6,0) {$5$};
\node[state,fill=blue!20,accepting] (6) at (6,2) {$6$};

\node[state,fill=blue!20] (7) at (0,2) {$7$};

\path[draw,thick,-latex] (1) edge node[above] {$a|\varepsilon$} (2);
\path[draw,thick,-latex] (2) edge node[above] {$\varepsilon|b$} (4);
\path[draw,thick,-latex] (4) edge node[above] {$\varepsilon|b$} (5);
\path[draw,thick,-latex] (5) edge node[right] {$b|\varepsilon$} (6);
\path[draw,thick,-latex] (2) edge[] node[left] {$a|\varepsilon$} (7);
\path[draw,thick,-latex] (7) edge[] node[above] {$b|\varepsilon$} (3);
\path[draw,thick,-latex] (3) edge[] node[right] {$\varepsilon|ba$} (2);

\path[draw,thick,-latex,dashed] (1) edge[loop below,pos=0.2] node[right]{$\varepsilon|\varepsilon$}();
\path[draw,thick,-latex,dashed] (2) edge[loop below,pos=0.2] node[right]{$\varepsilon|\varepsilon$}();
\path[draw,thick,-latex,dashed] (4) edge[loop below,pos=0.2] node[right]{$\varepsilon|\varepsilon$}();
\path[draw,thick,-latex,dashed] (6) edge[loop left,pos=0.5] node[left]{$\varepsilon|\varepsilon$}();
\path[draw,thick,-latex,dashed] (5) edge[loop below,pos=0.2] node[right]{$\varepsilon|\varepsilon$}();
\path[draw,thick,-latex,dashed] (3) edge[loop right] node[right]{$\varepsilon|\varepsilon$}();
\path[draw,thick,-latex,dashed] (7) edge[loop left] node[left]{$\varepsilon|\varepsilon$}();
\node(T1) at (4,1.5) {$T_1^\prime$};
\end{scope}
\begin{scope}[yshift=-4cm]
\node[state,initial, initial text=,fill=red!20] (1) at (0,0) {$1$};
\node[state,fill=red!20] (2) at (2,0) {$2$};
\node[state,fill=red!20] (3) at (2,2) {$3$};
\node[state,fill=red!20] (4) at (4,0) {$4$};
\node[state,fill=red!20] (5) at (6,0) {$5$};
\node[state,fill=red!20,accepting] (6) at (6,2) {$6$};

\path[draw,thick,-latex] (1) edge node[above] {$a|\varepsilon$} (2);
\path[draw,thick,-latex] (2) edge node[above] {$\varepsilon|a$} (4);
\path[draw,thick,-latex] (4) edge node[above] {$b|\varepsilon$} (5);
\path[draw,thick,-latex] (5) edge node[right] {$\varepsilon|a$} (6);
\path[draw,thick,-latex] (2) edge[bend left] node[left] {$ab|\varepsilon$} (3);
\path[draw,thick,-latex] (3) edge[bend left] node[right] {$\varepsilon|bb$} (2);

\node(T1) at (4,1.5) {$T_2$};
\end{scope}
\begin{scope}[xshift=8.5cm,yshift=-4cm]
\node[state,initial, initial text=,fill=red!20] (1) at (0,0) {$1$};
\node[state,fill=red!20] (2) at (2,0) {$2$};
\node[state,fill=red!20] (3) at (2,2) {$3$};
\node[state,fill=red!20] (4) at (4,0) {$4$};
\node[state,fill=red!20] (5) at (6,0) {$5$};
\node[state,fill=red!20,accepting] (6) at (6,2) {$6$};

\node[state,fill=red!20] (7) at (0,2) {$7$};

\path[draw,thick,-latex] (1) edge node[above] {$a|\varepsilon$} (2);
\path[draw,thick,-latex] (2) edge node[above] {$\varepsilon|a$} (4);
\path[draw,thick,-latex] (4) edge node[above] {$b|\varepsilon$} (5);
\path[draw,thick,-latex] (5) edge node[right] {$\varepsilon|a$} (6);
\path[draw,thick,-latex] (2) edge[] node[left] {$a|\varepsilon$} (7);
\path[draw,thick,-latex] (7) edge[] node[above] {$b|\varepsilon$} (3);
\path[draw,thick,-latex] (3) edge[] node[left] {$\varepsilon|bb$} (2);

\path[draw,thick,-latex,dashed] (1) edge[loop below,pos=0.2] node[right]{$\varepsilon|\varepsilon$}();
\path[draw,thick,-latex,dashed] (2) edge[loop below,pos=0.2] node[right]{$\varepsilon|\varepsilon$}();
\path[draw,thick,-latex,dashed] (4) edge[loop below,pos=0.2] node[right]{$\varepsilon|\varepsilon$}();
\path[draw,thick,-latex,dashed] (6) edge[loop left,pos=0.5] node[left]{$\varepsilon|\varepsilon$}();
\path[draw,thick,-latex,dashed] (5) edge[loop below,pos=0.2] node[right]{$\varepsilon|\varepsilon$}();
\path[draw,thick,-latex,dashed] (3) edge[loop right] node[right]{$\varepsilon|\varepsilon$}();
\path[draw,thick,-latex,dashed] (7) edge[loop left] node[left]{$\varepsilon|\varepsilon$}();

\node(T1) at (4,1.5) {$T_2$};
\end{scope}

\newcommand{\dessin}[2]{\(\textcolor{blue}{#1},\textcolor{red}{#2}\)}
\begin{scope}[xshift=1cm,yshift=-8cm]
\node[state,initial, initial text=,fill=green!20] (11) at (0,0) {\dessin{1}{1}};
\node[state,fill=green!20] (22) at (2,0) {\dessin{2}{2}};

\node[state,fill=green!20] (33) at (2,2) {\dessin{3}{3}};
\node[state,fill=green!20] (44) at (4,0) {\dessin{4}{4}};
\node[state,fill=green!20] (54) at (6,0) {\dessin{5}{4}};
\node[state,fill=green!20] (56) at (8,0) {\dessin{5}{6}};
\node[state,fill=green!20,accepting] (66) at (8,2) {\dessin{6}{6}};
\node[state,fill=green!20] (77) at (0,2) {\dessin{7}{7}};

\path[draw,thick,-latex] (11) edge node[above]{$\varepsilon|\varepsilon$}(22);
\path[draw,thick,-latex] (22) edge node[left]{$\varepsilon|\varepsilon$}(77);
\path[draw,thick,-latex] (77) edge node[above]{$\varepsilon|\varepsilon$}(33);
\path[draw,thick,-latex] (33) edge node[right]{$ba|bb$}(22);
\path[draw,thick,-latex] (22) edge node[above]{$b|a$}(44);
\path[draw,thick,-latex] (44) edge node[above]{$b|\varepsilon$}(54);

\path[draw,thick,-latex] (54) edge node[above]{$\varepsilon|\varepsilon$}(56);

\path[draw,thick,-latex] (56) edge node[right]{$\varepsilon|a$}(66);

\path[draw,thick,-latex] (11) edge[loop below,pos=0.2] node[right]{$\varepsilon|\varepsilon$}();

\path[draw,thick,-latex,] (22) edge[loop below,pos=0.2] node[right]{$\varepsilon|\varepsilon$}();
\path[draw,thick,-latex,] (44) edge[loop below,pos=0.2] node[right]{$\varepsilon|\varepsilon$}();
\path[draw,thick,-latex,] (66) edge[loop left,pos=0.5] node[left]{$\varepsilon|\varepsilon$}();
\path[draw,thick,-latex,] (54) edge[loop below,pos=0.2] node[right]{$\varepsilon|\varepsilon$}();
\path[draw,thick,-latex,] (56) edge[loop below,pos=0.2] node[right]{$\varepsilon|\varepsilon$}();
\path[draw,thick,-latex,] (33) edge[loop right] node[right]{$\varepsilon|\varepsilon$}();
\path[draw,thick,-latex,] (77) edge[loop left] node[left]{$\varepsilon|\varepsilon$}();

\node[text width= 3cm](T1T2) at (11,1.5) {Accessible part of the product};
\end{scope}

\end{tikzpicture}}
	\caption{Reduction from \(\Cal{T}_1\) and \(\Cal{T}_2\), with \(\dev(R_{(\Cal{T}'_1 \circ \Cal{T}'_2)}) = +\infty\).}
	\label{fig:reduction-exemple}
\end{figure}



\section{Lower Bounds}\label{sec:lowerbounds}
This section is devoted to proving the lower bounds stated in Theorems~\ref{thm-BoundedNFT}-\ref{thm-Kexact}.
We establish the results separately for each complexity class,
by reductions from standard complete problems for NL (Subsection~\ref{subsec:NL}), co-NP (Subsection~\ref{subsec:co-NP}), and DP (Subsection~\ref{subsec:DP}), respectively.

\subsection{NL-hardness} 
\label{subsec:NL}
We recall the following NL-complete problem~\cite[Theorem 8.4]{Papadimitriou94}.
\PB{Directed Reachability Problem}
{A directed graph \(G=(V,E)\) and two vertices \(s,t \in V\).}
{\true{} if there exists a path from \(s\) to \(t\) in \(G\), \false{} otherwise.}

\noindent
Note that the following reductions for the NL class produce trim transducers without \((\varepsilon,\varepsilon)\)-transitions.
This shows that the hardness is intrinsic to the problems,
and does not arise from auxiliary tasks such as eliminating \(\varepsilon\)-transitions or trimming the transducer.

In both reductions, we reduce the directed reachability problem
by embedding the input graph into the transition structure of a transducer,
adding an initial state and a final state connected to all vertices
to ensure that the resulting transducer is trim.
We then add a few transitions that introduce mismatches
when a specific path exists.
This actually results in a reduction from the \emph{complement} of the reachability problem,
which is equivalent since, by the Immerman–Szelepcsényi Theorem,  NL = co-NL~\cite{Immerman88,Szelepcsenyi88}.

\begin{proposition}
	\label{prop:bounded-hardness}
	The Bounded Deviation Problem is NL-hard.
\end{proposition}
\begin{proof}
	Let \(G=(V,E)\) and \(s,t \in V\) be an instance of the Directed Reachability Problem.
	We construct a trim NFT \(\Cal{T} = (V \cup \{q_i,q_f\},\{a,b\}, \{q_i\}, \{q_f\}, \Delta)\),
	where \(\Delta\) is defined as follows:
	\[
		\Delta = \{
		(q_i,a,a,v), (v,a,a,q_f) \mid v \in V\}
		\cup \{
		(u,a,a,v) \mid (u,v) \in E
		\}
		\cup 
		\{
		(t,a,b,s)
		\}.
	\]
	This construction is realizable in \logspace.
	Note that all the transitions read and output the letter \(a\),
	except the transition \((t,a,b,s)\), which introduces a mismatch.
	Assume first that there is a path from \(s\) to \(t\) in \(G\).
	This path induces a corresponding run in \(\Cal{T}\) from \(s\) to \(t\),
	which forms a cycle once concatenated with transition \((t,a,b,s)\).
	Iterating this cycle arbitrarily many times yields runs with arbitrarily many mismatches.
	As a consequence, \(\dev(R_{\Cal{T}}) = \infty\).
	Conversely, assume  there is no path from \(s\) to \(t\) in \(G\).
	Then the transition \((t,a,b,s)\) occurs at most once in each run of \(\Cal{T}\),
	thus \(\dev(R_{\Cal{T}}) \leq 1\).
	Therefore, \(\dev(R_{\Cal{T}}) < \infty\) if and only if \(t\) is not reachable from \(s\) in \(G\),
	which concludes the proof.
\end{proof}

\begin{proposition}
	\label{prop:fixed-hardness}
	The Threshold-bounded Deviation Problem is NL-hard for all fixed \(k \geq 1\).
\end{proposition}
\begin{proof}
	Let \(G=(V,E)\) and \(s,t \in V\) be an instance of the Directed Reachability Problem.
	We construct a trim NFT \(\Cal{T} = (V \cup \{q_i,q_f\},\{a,b\}, \{q_i\}, \{q_f\}, \Delta)\),
	where \(\Delta\) is defined as follows:
	\[
		\Delta = \{
		(q_i,a,a,v), (v,a,a,q_f) \mid v \in V\}
		\cup \{
		(u,a,a,v) \mid (u,v) \in E
		\}
		\cup
		\{
		(q_i,a^k,b^k,s), (t,a,b,q_f)
		\}.
	\]
	All the transitions read and output the letter \(a\),
	except the transitions \((q_i,a^k,b^k,s)\) and \((t,a,b,q_f)\),
	which introduce  \(k\) and \(1\) mismatches respectively.
	If \(t\) is reachable from \(s\) in \(G\),
	there exists a run of \(\Cal{T}\) that starts with \((q_i,a^k,b^k,s)\),
	then follows a path of \(G\) from \(s\) to \(t\), and concludes with the transition \((t,a,b,q_f)\).
	This run contains exactly \(k+1\) mismatches, hence \(\dev(R_{\Cal{T}}) > k\).
	Conversely, if \(t\) is not reachable from \(s\), then 
	no run can contain both \((q_i,a^k,b^k,s)\)
	and \((t,a,b,q_f)\).
	Hence, every run contains at most \(k\) mismatches, thus \(\dev(\Cal{T}) \leq k\).
	Therefore, \(\dev(R_{\Cal{T}}) \leq k\) if and only if \(t\) is not reachable from \(s\) in \(G\),
	which concludes the proof.
\end{proof}

\subsection{co-NP hardness} 
\label{subsec:co-NP}
We recall the following NP-complete problem~\cite[Theorem 6.1]{Papadimitriou94}.

\PB{3-SAT}
{A Boolean formula \(\varphi = \bigwedge_i c_i\) where each clause \(c_i\) is a disjunction of three literals.}
{\true{} if \(\varphi\) is satisfiable, \false{} otherwise.}

\begin{proposition}
	\label{prop:threshold-bounded-hardness}
	The Threshold-bounded Deviation Problem is co-NP-hard.
\end{proposition}

\begin{proof}
	Let \(\varphi\) be a 3-SAT instance on \(n\) variables composed of \(m\) clauses:
	\[
		\varphi(x_1,\ldots,x_n) = \bigwedge\limits_{i = 1}^m  \big(\ell_{i}^1 \vee \ell_{i}^2 \vee \ell_{i}^3\big)\text{, where } \ell_{i}^j \text{ is either } x_k \text{ or } \lnot x_k, \text{ with } 1 \leq k \leq n.
	\]
	We construct an NFT \(\Cal{T}\) such that \(\dev(R_{\Cal{T}}) > n \cdot (m+1)-1\), if and only if the 3-SAT instance has a solution.
	The construction of $\Cal{T}$ is based on a sequential concatenation of gadgets.
	Intuitively, for each clause \(c_i\) we construct an NFT \(\Cal{T}_i\)
	that reads a word \(u \in \{0,1\}^n\) encoding a valuation,
	accepts if the valuation satisfies the clause,
	and produces its bitwise negation.
	To evaluate all the clauses,
	these gadgets are then sequentially combined with an initial shift of minus \(n\).
	The proof of correctness then relies on the fact that the input of a gadget produces \(n\) mismatches
	with the output of its predecessor
	if and only if they both read the same valuation. An example of a clause gadget can be found in~\cref{fig:gadget}. The different gadgets and their combination is given in~\cref{fig:co-np-hardness-reduction-exemple} for two clauses and four variables.
	
	\begin{figure}
		\centering
		\resizebox{0.6\textwidth}{!}{

\begin{tikzpicture}

\node[state,fill=red!20](0T) at (0,0) {\(0,\top_1\)};
\node[state,fill=red!20](1T) at (3,0) {\(1,\top_1\)};
\node[state,fill=red!20](2T) at (6,0) {\(2,\top_1\)};
\node[state,fill=red!20](3T) at (9,0) {\(3,\top_1\)};
\node[state,fill=red!20,accepting](4T) at (12,0) {\(4,\top_1\)};

\node[state,fill=red!20,initial,initial text=](0b) at (0,2) {\(0,\bot_1\)};
\node[state,fill=red!20](1b) at (3,2) {\(1,\bot_1\)};
\node[state,fill=red!20](2b) at (6,2) {\(2,\bot_1\)};
\node[state,fill=red!20](3b) at (9,2) {\(3,\bot_1\)};
\node[state,fill=red!20](4b) at (12,2) {\(4,\bot_1\)};

\path[draw,thick,-latex] (0T) edge node[above] {\(1|0\)} node[below] {\(0|1\)}(1T);
\path[draw,thick,-latex] (1T) edge node[above] {\(1|0\)} node[below] {\(0|1\)}(2T);
\path[draw,thick,-latex] (2T) edge node[above] {\(1|0\)} node[below] {\(0|1\)}(3T);
\path[draw,thick,-latex] (3T) edge node[above] {\(1|0\)} node[below]
     {\(0|1\)}(4T);
     
\path[draw,thick,-latex] (0b) edge node[above] {\(1|0\)} node[below] {\(0|1\)}(1b);
\path[draw,thick,-latex] (1b) edge node[above] {\(1|0\)} node[below] {\(0|1\)}(2b);
\path[draw,thick,-latex] (2b) edge node[above] {\(1|0\)} node[below] {\(0|1\)}(3b);
\path[draw,thick,-latex] (3b) edge node[above] {\(1|0\)} node[below] {\(0|1\)}(4b);

\path[draw,thick,-latex] (0b) edge node[left] {\(1|0 \ \)} (1T);
\path[draw,thick,-latex] (1b) edge node[left] {\(0|1 \ \)} (2T);
\path[draw,thick,-latex] (3b) edge node[left] {\(0|1 \ \)} (4T);
\end{tikzpicture}}
		\caption{Gadget for \((x_1 \vee \lnot x_2 \vee \lnot x_4)\).}
		\label{fig:gadget}
	\end{figure}

	\proofsubparagraph{Clause gadgets.}
	For each clause \(c_i\), the NFT \(\Cal{T}_i\) is defined as \((Q_i,\{0,1\}, \{s_{i}\}, \{f_{i}\}, \Delta_i)\),
	where \(Q_i = [0,n] \times \{\top_i,\bot_i\}\), \(s_i = (0,\bot_i)\), \(f_i = (n,\top_i)\), and \(\Delta_i\) contains the following quadruplets:
	\begin{enumerate}
		\item For all \(j < n\) and \(b \in \{0,1\}\), \(\big((j,\top_i), b, 1-b, (j+1,\top_i)\big) \in \Delta_i\).
		\item For all \(j < n\) and \(b \in \{0,1\}\), \(\big((j,\bot_i), b, 1-b, (j+1,\bot_i)\big) \in \Delta_i\).
		\item For all \(k\) such that the literal \(x_k\) occurs in \(c_i\), \(\big((k-1,\bot_i), 1, 0, (k,\top_i)\big) \in \Delta_i\).
		\item For all \(k\) such that the literal \(\lnot x_k\) occurs in \(c_i\), \(\big((k-1,\bot_i), 0, 1, (k,\top_i)\big) \in \Delta_i\).
	\end{enumerate}
	Note that every accepting run of \(\Cal{T}_i\) is of length exactly \(n\),
	and outputs the negation of its input.
	Furthermore, the transitions of type~3 and~4 above are the only ones that move from the
	\(\bot_i\)-states to the \(\top_i\)-states, and thus are necessary to reach the final state.
	They are enabled exactly when the corresponding literal makes the clause \(c_i\) true,
	which ensures that the accepting runs of \(\Cal{T}_i\) encode valuations satisfying \(c_i\).

	\proofsubparagraph{Boundary gadgets.}
	We construct two gadgets NFT \(\Cal{T}_{\Rm{init}} = (Q_{\Rm{init}},\{0,1\}, \{s_{\Rm{init}}\}, \{f_{\Rm{init}}\}, \Delta_\Rm{init})\) 
	and \(\Cal{T}_{\Rm{final}} = (Q_{\Rm{final}},\{0,1\}, \{s_{\Rm{final}}\}, \{f_{\Rm{final}}\}, \Delta_\Rm{final})\), where:
	\begin{itemize}
	\item \makebox[5cm][l]{\(Q_{\Rm{init}}=[0,n]\times\{\Rm{init}\}\),} \(Q_{\Rm{final}}=[0,n]\times\{\Rm{final}\}\),
	\item \makebox[5cm][l]{\(s_{\Rm{init}}=(0,\Rm{init})\),} \(s_{\Rm{final}}=(0,\Rm{final})\), 
	\item \makebox[5cm][l]{\(f_{\Rm{init}}=(n,\Rm{init})\),} \(f_{\Rm{final}}=(n,\Rm{final})\), 
	\item \( \Delta_\Rm{init}=\{\big((i,\Rm{init}),\varepsilon,b,(i+1,\Rm{init})\big)\mid 0\leq i<n \text{ and }b\in\{0,1\}\}\),\\
		\( \Delta_\Rm{final}=\{\big((i,\Rm{final}),b,\varepsilon,(i+1,\Rm{final})\big)\mid 0\leq i<n \text{ and }b\in\{0,1\}\}\).
\end{itemize}
Note that \(R_{\Cal{T}_{\Rm{init}}}=\{\varepsilon\}\times\{0,1\}^n\) and 
 \(R_{\Cal{T}_{\Rm{final}}}=\{0,1\}^n\times\{\varepsilon\}\).

	\proofsubparagraph{Global construction.}
	We construct an NFT \(\Cal{T}\) that recognizes the concatenation of the relations of all the gadgets: \(R_{\Cal{T}} = R_{\Cal{T}_{\Rm{init}}} \cdot R_{\Cal{T}_{1}} \cdot \ldots \cdot R_{\Cal{T}_{m}} \cdot R_{\Cal{T}_{\Rm{final}}}\).
	Since each gadget has a unique initial state with no incoming transition
	and a unique final state with no outgoing transition,
	this concatenation is achieved by merging the final state of each gadget
	with the initial state of the following one~: 
	\(f_{\Rm{init}}\) with \(s_1\), \(f_i\) with \(s_{i+1}\) for all \(1 \leq i < m\), and \(f_m\) with \(s_{\Rm{final}}\).
	The initial state of \(\Cal{T}\) is \(s_{\Rm{init}}\) and the final state is \(f_{\Rm{final}}\).
	The construction is then polynomial, the number of states of \(\Cal{T}\) being \((n+1) \cdot (2m+2)-(m+1)=(2n+1)\cdot(m+1)\).
	
	\proofsubparagraph{Correctness.}
	We prove that \(\varphi\) is satisfiable if and only if
	there exists \((y,z)\) in \(R_{\Cal{T}}\) such that \(d(y,z) \geq n \cdot (m+1)\).
	Suppose that \(\varphi\) is satisfiable and let \(\nu\) be a valuation satisfying it.
	Let \(u = \nu(x_1)\cdots\nu(x_n)\) be the word of length \(n\) encoding this valuation and let \(v =  \lnot\nu(x_1)\cdots\lnot\nu(x_n)\) be its negation.
	Since \(\nu\) is a valuation satisfying \(\varphi\), it satisfies each of its clauses, hence for all \(i \in [1,m]\), \((u,v) \in R_{\Cal{T}_i}\).
	Furthermore, \((\varepsilon,v) \in R_{\Cal{T}_{\Rm{init}}}\) and \((u,\varepsilon) \in R_{\Cal{T}_{\Rm{final}}}\).
	Consequently, \((u^{m+1},v^{m+1}) \in R_{\Cal{T}}\), with \(d(u^{m+1},v^{m+1}) = |u| \cdot (m+1) = n \cdot (m+1)\).

	Now let \((y,z) \in R_{\Cal{T}}\), since \(R_{\Cal{T}} = R_{\Cal{T}_{\Rm{init}}} \cdot R_{\Cal{T}_{1}} \cdot \ldots \cdot R_{\Cal{T}_{m}} \cdot R_{\Cal{T}_{\Rm{final}}}\), we can decompose \(y\) and \(z\) into \(u_1\cdots u_m u_{\Rm{final}}\) and \(v_{\Rm{init}}v_1\cdots v_{m}\) respectively, with for all \(i \in [1,m]\), \(|v_{\Rm{init}}| = |u_{\Rm{final}}| = |u_i| = |v_i| = n\).
	By definition of \(R_{\Cal{T}_i}\), \(u_i\) describes a valuation satisfying \(c_i\) and \(v_i\) is the negation of \(u_i\).
	If \(d(y,z) \geq n \cdot (m+1)\) while \(|y| = |z| = n \cdot (m+1)\), all positions of \(y\) and \(z\) mismatch.
	Consequently, \(u_1\) is the negation of \(v_{\Rm{init}}\), for all \(i \in [2,m]\) \(u_i\) is the negation of \(v_{i-1}\), and \(u_{\Rm{final}}\) is the negation of \(v_m\).
	Therefore, \(u_1 = u_2 = \ldots = u_m = u_{\Rm{final}}\), and they all encode the same valuation satisfying \(c_1,c_2,\ldots,c_m\) and consequently satisfying \(\varphi\).
\end{proof}

\begin{figure}
	\centering
	\resizebox{0.95\textwidth}{!}{
{\Large
\begin{tikzpicture}


  \draw[fill=blue!10] (-0.75,-1.75) rectangle +(13.5,4.5);
  \node () at (6,-1.2) {\(\mathcal{T}_1\) for \((x_1\wedge \neg x_2 \wedge \neg
    x_3)\)};
  
\node[state,fill=blue!20](0T1) at (0,0) {};
\node[state,fill=blue!20](1T1) at (3,0) {};
\node[state,fill=blue!20](2T1) at (6,0) {};
\node[state,fill=blue!20](3T1) at (9,0) {};
\node[state,fill=blue!20](4T1) at (12,0) {};

\node[state,fill=blue!20](0b1) at (0,2) {};
\node[state,fill=blue!20](1b1) at (3,2) {};
\node[state,fill=blue!20](2b1) at (6,2) {};
\node[state,fill=blue!20](3b1) at (9,2) {};
\node[state,fill=blue!20](4b1) at (12,2) {};

\path[draw,thick,-latex] (0T1) edge node[above] {\(1|0\)} node[below] {\(0|1\)}(1T1);
\path[draw,thick,-latex] (1T1) edge node[above] {\(1|0\)} node[below] {\(0|1\)}(2T1);
\path[draw,thick,-latex] (2T1) edge node[above] {\(1|0\)} node[below] {\(0|1\)}(3T1);
     
\path[draw,thick,-latex] (0b1) edge node[above] {\(1|0\)} node[below] {\(0|1\)}(1b1);
\path[draw,thick,-latex] (1b1) edge node[above] {\(1|0\)} node[below] {\(0|1\)}(2b1);
\path[draw,thick,-latex] (2b1) edge node[above] {\(1|0\)} node[below] {\(0|1\)}(3b1);
\path[draw,thick,-latex] (3b1) edge node[above] {\(1|0\)} node[below] {\(0|1\)}(4b1);

\path[draw,thick,-latex] (0b1) edge node[left] {\(1|0 \ \)} (1T1);
\path[draw,thick,-latex] (1b1) edge node[left] {\(0|1 \ \)} (2T1);
\path[draw,thick,-latex] (2b1) edge node[left] {\(0|1 \ \)} (3T1);


\begin{scope}[xshift=12cm,yshift=-2cm]

  \draw[fill=red!10] (-0.75,-1.75) rectangle +(13.5,4.5);
  \node () at (6,-1.2) {\(\mathcal{T}_2\) for \((x_2\wedge \neg x_3 \wedge x_4)\)};
  
\node[state,fill=red!20](0T) at (0,0) {};
\node[state,fill=red!20](1T) at (3,0) {};
\node[state,fill=red!20](2T) at (6,0) {};
\node[state,fill=red!20](3T) at (9,0) {};
\node[state,fill=red!20](4T) at (12,0) {};

\node[state,fill=red!20](0b) at (0,2) {};
\node[state,fill=red!20](1b) at (3,2) {};
\node[state,fill=red!20](2b) at (6,2) {};
\node[state,fill=red!20](3b) at (9,2) {};
\node[state,fill=red!20](4b) at (12,2) {};

\path[draw,thick,-latex] (3T1) edge node[above] {\(1|0\)} node[below]
     {\(0|1\)}(4T1);

\path[draw,thick,-latex] (0T) edge node[above] {\(1|0\)} node[below] {\(0|1\)}(1T);
\path[draw,thick,-latex] (1T) edge node[above] {\(1|0\)} node[below] {\(0|1\)}(2T);
\path[draw,thick,-latex] (2T) edge node[above] {\(1|0\)} node[below] {\(0|1\)}(3T);
\path[draw,thick,-latex] (3T) edge node[above] {\(1|0\)} node[below]
     {\(0|1\)}(4T);
     
\path[draw,thick,-latex] (0b) edge node[above] {\(1|0\)} node[below] {\(0|1\)}(1b);
\path[draw,thick,-latex] (1b) edge node[above] {\(1|0\)} node[below] {\(0|1\)}(2b);
\path[draw,thick,-latex] (2b) edge node[above] {\(1|0\)} node[below] {\(0|1\)}(3b);
\path[draw,thick,-latex] (3b) edge node[above] {\(1|0\)} node[below] {\(0|1\)}(4b);

\path[draw,thick,-latex] (1b) edge node[left] {\(1|0 \ \)} (2T);
\path[draw,thick,-latex] (2b) edge node[left] {\(0|1 \ \)} (3T);
\path[draw,thick,-latex] (3b) edge node[left] {\(1|0 \ \)} (4T);

\end{scope}

\begin{scope}[xshift=0cm,yshift=4cm]

  \draw[fill=green!10] (-0,-1) rectangle +(15,2);
  \node () at (14,0) {\(\mathcal{T}_{\rm init}\)};
  
\node[state,fill=green!20](0Ti) at (12,0) {};
\node[state,fill=green!20](1Ti) at (9,0) {};
\node[state,fill=green!20](2Ti) at (6,0) {};
\node[state,fill=green!20](3Ti) at (3,0) {};

\path[draw,thick,-latex] (0Ti) edge node[above] {\(\varepsilon|0\)} node[below] {\(\varepsilon|1\)}(1Ti);
\path[draw,thick,-latex] (1Ti) edge node[above] {\(\varepsilon|0\)} node[below] {\(\varepsilon|1\)}(2Ti);
\path[draw,thick,-latex] (2Ti) edge node[above] {\(\varepsilon|0\)} node[below] {\(\varepsilon|1\)}(3Ti);
\path[draw,thick,-latex] (3Ti) edge node[left,pos=0.15] {\(\varepsilon|0,\ \varepsilon|1\)}(0b1);

\end{scope}
\begin{scope}[xshift=12cm,yshift=-4.75cm]

  \draw[fill=yellow!10] (-3.5,-.75) rectangle +(15,1.5);
  \node () at (-2,0) {\(\mathcal{T}_{\rm final}\)};
  
\node[state,fill=yellow!20](1Tf) at (9,0) {};
\node[state,fill=yellow!20](2Tf) at (6,0) {};
\node[state,fill=yellow!20](3Tf) at (3,0) {};
\node[state,fill=yellow!20,accepting](4Tf) at (0,0) {};

\path[draw,thick,-latex] (4T) edge node[right,pos=0.4] {\(\ 0|\varepsilon,\ 1|\varepsilon\)}(1Tf);
\path[draw,thick,-latex] (1Tf) edge node[above] {\(0|\varepsilon\)} node[below] {\(1|\varepsilon\)}(2Tf);
\path[draw,thick,-latex] (2Tf) edge node[above] {\(0|\varepsilon\)} node[below] {\(1|\varepsilon\)}(3Tf);
\path[draw,thick,-latex] (3Tf) edge node[above] {\(0|\varepsilon\)} node[below] {\(1|\varepsilon\)}(4Tf);

\end{scope}

\end{tikzpicture}}}
	\caption{An NFT \(\Cal{T}\) for a SAT formula with 2 clauses and 4 variables.}
	\label{fig:co-np-hardness-reduction-exemple}
\end{figure}
\label{ssub:completness}

\subsection{DP-Hardness} 
\label{subsec:DP}
\label{ssec:dp_hardness}

A canonical example of a DP problem is SAT-UNSAT~\cite[S.~2,~Lemma~1]{Papadimitriou1984}:

\PB{SAT-UNSAT Problem}
{Two Boolean formulas \(\varphi_1,\varphi_2\).}
{\true{} if \(\varphi_1\) is satisfiable and \(\varphi_2\) is not satisfiable, \false{} otherwise.}

\begin{proposition}
	\label{prop:exact-hardness}
	The Exact Deviation Problem is DP-hard.
\end{proposition}

\begin{proof}
	Consider an instance of SAT-UNSAT given by two propositional formulas \(\varphi_1\) and \(\varphi_2\).
	We reuse the construction of the proof of the \cref{prop:threshold-bounded-hardness} to obtain two NFTs \(\Cal{T}_1\), \(\Cal{T}_2\) and two positive integers \(k_1\), \(k_2\) such that for \(i \in \{1,2\}\), \(|u| = |v| = k_i\) for all \((u,v) \in R_{\Cal{T}_i}\), and:
	\[
		\begin{cases}
			\dev(R_{\Cal{T}_i}) \leq k_i -1 &\text{ if } \varphi_i \text{ is not satisfiable,} \\
			\dev(R_{\Cal{T}_i}) = k_i &\text{ if } \varphi_i \text{ is satisfiable.}
		\end{cases}
	\]
	We then create \(\Cal{T}'_1\) by concatenating \(k_2\) copies of \(\Cal{T}_1\) together, i.e., \(R_{\Cal{T}'_1} = R_{\Cal{T}_1}^{k_2}\).
	Therefore:
	\[
		\begin{cases}
			\dev(R_{\Cal{T}'_1}) \leq (k_{1} - 1)k_{2} &\text{ if } \varphi_1 \text{ is not satisfiable,} \\
			\dev(R_{\Cal{T}'_1}) = k_{1}k_{2} &\text{ if } \varphi_1 \text{ is satisfiable.}
		\end{cases}
	\]
	Consider \(\Cal{T}\) the concatenation of \(\Cal{T}'_1\) and \(\Cal{T}_2\).
	By construction, for all \((u,v) \in R_{\Cal{T}}\), \(|u| = |v| = k_1k_2 + k_2\) and \(\Cal{T}\) is a bounded NFT.
	Furthermore, \(\varphi_1\) and \(\varphi_2\) are both satisfiable if and only if there exists \((u,v) \in R_{\Cal{T}}\) such that \(d(u,v) = k_1k_2 + k_2\).
	Note also that if \(\varphi_1\) is not satisfiable, for all \((u,v) \in R_{\Cal{T}'_1}\), there will be at most \((k_1-1)k_2\) mismatches from \(\Cal{T}'_1\) and at most \(k_2\) mismatches from \(\Cal{T}_2\), there will be at most \(k_1k_2\) mismatches.
	The SAT-UNSAT instance is true if and only if the bound of \(\Cal{T}\) is greater than \(k_1k_2\), and strictly less than \(k_1k_2 + k_2\).

	We conclude by reducing this interval to a single value in the following way.
	Let \(\Cal{T}_3\) be an NFT with two states and one transition such that \(R_{\Cal{T}_3} = \{(0^{k_2-1}, 1^{k_2-1})\}\).
	Consider \(\Cal{S}\) the NFT such that \(\Cal{S}\) is the concatenation of \(\Cal{T}'_1\) with the non-deterministic choice between \(\Cal{T}_2\) and \(\Cal{T}_3\).
	We have \(R_{\Cal{S}} = R_{\Cal{T}'_1} \cdot (R_{\Cal{T}_2} \cup R_{\Cal{T}_3})\).
	In the end, the SAT-UNSAT instance is true if and only if the bound of \(\Cal{S}\) is exactly \(k_1k_2 +k_2-1\).
\end{proof}

\begin{figure}
	\centering
	\resizebox{0.6\textwidth}{!}{
\begin{tikzpicture}


  \draw[fill=blue!10,thick] (-1,-1) rectangle +(8cm,2.5cm);
  \node (T1prime) at (3,1.2) {\(\mathcal{T}_1^\prime\)};

\node[state, rectangle,thick,fill=blue!20,minimum width=1.5cm,minimum height=1.5cm] (T1a) at (0,0) {\(\mathcal{T}_1\)};
\node[state, rectangle,thick,fill=blue!20,minimum width=1.5cm,minimum height=1.5cm] (T1b) at (3,0) {\(\mathcal{T}_1\)};
  \node (dot) at (4.5,0) {\(\cdots\)};
\node[state, rectangle,thick,fill=blue!20,minimum width=1.5cm,minimum height=1.5cm] (T1c) at (6,0) {\(\mathcal{T}_1\)};

\path[thick,draw,-latex] (T1a) -- (T1b);
\path[thick,draw,-latex] (T1b) -- (dot);
\path[thick,draw,-latex] (dot) -- (T1c);

\node[state, rectangle,thick,fill=red!20,minimum width=1.5cm,minimum
  height=1.5cm] (T2) at (9,1) {\(\mathcal{T}_2\)};

\node[state, rectangle,thick,fill=green!20,minimum width=1.5cm,minimum height=1.5cm] (T3) at (9,-1) {\(\mathcal{T}_3\)};
\path[thick,draw,-latex] (T1c) -- (T2);
\path[thick,draw,-latex] (T1c) -- (T3);
\end{tikzpicture}}
	\caption{Resulting NFT \(\Cal{S}\) for SAT-UNSAT to Exact Deviation Problem reduction.}
	\label{fig:dp-hardness-reduction-exemple}
\end{figure}



\section{Upper Bounds} 
\label{sec:upperBounds}
This section is devoted to proving the upper bounds stated
in Theorems~\ref{thm-BoundedNFT} and~\ref{thm-KboundedNFT}.
In Subsection~\ref{sub:bounded_nft} we show that the Bounded Deviation Problem is in NL 
and that whenever an NFT \(\cal{T}\) is bounded,
the bound is at most quadratic in \(|\cal{T}|\) (Proposition~\ref{prop:BoundedNFTHigh}).
Then, Subsection~\ref{ssub:k_bounded_nft}
proves the upper bound for the Threshold-bounded Deviation Problem.
In the following section, we construct indiscriminately NL or co-NL algorithms, as they are proven to be equivalent~\cite{Immerman88,Szelepcsenyi88}.

\subsection{Bounded Deviation Problem} 
\label{sub:bounded_nft}

This subsection is dedicated to prove the following upper bound:
\begin{proposition}\label{prop:BoundedNFTHigh}
The Bounded Deviation Problem is in NL.
Moreover, if \(\mathcal{T}\) is bounded,
then its deviation is in \(O(|\mathcal{T}|^2)\).
\end{proposition}
We prove Proposition~\ref{prop:BoundedNFTHigh} in two steps.
First, we show in Lemma~\ref{lemma:bounded_equal_properties}
that bounded NFTs are characterized by a conjunction of two properties:
being length-preserving,
and having, for each cyclic run,
input and output words that are conjugate with respect to a state-dependent shift.
Second, we show that both properties are decidable in NL:
Lemmas~\ref{lemma:complexity_u_equal_v_for_small_run} and~\ref{lemma:complexity_u_equal_v_for_cycle}
imply together that being length-preserving is decidable in NL,
and Lemma~\ref{lemma:complexity_conjugate_cycle} provides an NL algorithm
for checking loop conjugacy.

To formalize the notion of state-dependent shift used in the loop conjugacy property,
we first establish the following lemma.
It shows that, in a length-preserving NFT,
all runs leading to a given state \(p\) share the same shift \(s_p\),
which we call the \emph{shift of the state \(p\)}.

\begin{lemma} 
\label{lemma:bounded_imply_one_shift_by_state}
	Let \(\Cal{T} = (Q,\Sigma, Q_{i}, Q_{f}, \Delta)\) be a length-preserving NFT.
	For all states \(p \in Q\), there exists \(s_p \in \mathbb{Z}\)
	such that every run \(\rho \in \Delta^*\)
	from \(q\in Q_i\) to \(p\) satisfies \(\shift(\rho) = s_p\).
	Moreover, \(|s_p| \leq \min(\smax(\Cal{T}) \cdot |Q|,|\Cal{T}|)\).
\end{lemma}

\begin{proof}
	Let \(p \in Q\) be a state of \(\Cal{T}\) and let \(\rho\) and \(\rho'\) be two runs of \(\Cal{T}\), from initial states 
	to \(p\).
	Let \(\rho_{final}\) be a run of \(\Cal{T}\) from \(p\) to some final state.
	Since \(\Cal{T}\) is length-preserving, we have \(\shift(\rho\rho_{final}) = 0\) and \(\shift(\rho'\rho_{final}) = 0\).
	Hence, \( \shift(\rho\rho_{final}) = \shift(\rho'\rho_{final})\), and  \(\shift(\rho) = \shift(\rho') = s_p\).
	
	Regarding the size of \(s_p\), let us consider a smallest initial run \(\rho\) to \(p\).
	Since \(|\rho|\) is smaller than \(|Q|\) and each transition generates a shift smaller than \(\smax(\Cal{T})\), \(s_p\) is smaller than \(\smax(\Cal{T}) \cdot |Q|\).
	Moreover, \(\rho\) uses each transition at most once,
	hence the shift \(s_p\) is smaller than the total representation of \(\Cal{T}\), so \(s_p \leq |\Cal{T}|\).
\end{proof}

\begin{remark}\label{rmk:QuadBound}
We discuss here the double bound on \(s_p\).
These bounds are in general incomparable,
and each can be advantageous depending on the transducer \(\Cal{T}\). 
If \(\Cal{T}\) has many transitions compared to its number of states,
it is more efficient to rely on the \(\smax(\Cal{T}) \cdot |Q|\).
On the other hand, if \(\Cal{T}\) has few transitions but large shift values,
then \(\smax(\Cal{T}) \cdot |Q|\) might be close to \(|\Cal{T}|^2\), 
in which case the bound \(|\Cal{T}|\) is preferable.
Notably, we rely on this later bound to obtain a generic quadratic bound on the deviation of \(\Cal{T}\) in~\cref{lemma:QuadBound_imply_not_properties}.
\end{remark}

We now formally express the conditions for an NFT to be bounded.
\begin{lemma}\label{lemma:bounded_equal_properties}
    Let \(\Cal{T} = (Q,\Sigma, Q_{i}, Q_{f}, \Delta)\) be an NFT.
    Then \(\Cal{T}\) is bounded if and only if both:
    \begin{enumerate}
    	\item \(\Cal{T}\) is length-preserving, and\label{item:bounded_equal_properties_1}
    	\item for every run \(\rho \in \Delta^*\) from \(p \in Q\) to itself over some \((u,v)\), \(u\) is conjugate to \(v\) by \(s_p\).\label{item:bounded_equal_properties_2}
    \end{enumerate}
    Moreover, these conditions imply that the deviation of \(\mathcal{T}\)
    is in \(O(|\mathcal{T}|^2)\).
\end{lemma}

\noindent
To prove \cref{lemma:bounded_equal_properties},
we first show that the conditions are necessary:
if \(\Cal{T}\) is bounded then both conditions hold
(the first follows directly from the definition, and the second from \cref{lemma:bounded_imply_conjugate_cycle}).
Then, to prove that the conditions are sufficient and provide the stated bound,
we show that if the deviation of \(\Cal{T}\) exceeds a quadratic threshold,
then at least one of the conditions must fail (Lemma~\ref{lemma:QuadBound_imply_not_properties}).

\subparagraph*{Necessary conditions for boundedness} 

\noindent
First, notice that since words of different lengths have infinite distance,
a bounded NFT is necessarily length-preserving.
Therefore, the first condition of \cref{lemma:bounded_equal_properties} is necessary.
As a first step towards proving that the second condition is necessary,
we establish the following consequence of length-preservation.

\begin{lemma}\label{lemma:length_preserving_imply_u_equal_v}
	Let \(\Cal{T} = (Q,\Sigma, Q_{i}, Q_{f}, \Delta)\) be a length-preserving NFT. Then for all runs \(\rho \in \Delta^*\) from \(q \in Q\) to itself over some \((u,v)\), \(|u| = |v|\).
\end{lemma}

\begin{proof}
Let \(\rho\) be a run of \(\Cal{T}\) from \(p \in Q\) to itself over some \((u,v)\).
	Let \(\rho_{\Rm{init}}\) be an initial run of \(\Cal{T}\) to \(p\)
	and let \(\rho_{\Rm{final}}\) be a final run of \(\Cal{T}\) from \(p\).
	Since \(\rho\) is a loop, \(\rho_{\Rm{init}}\rho^2\rho_{\Rm{final}}\) is a valid run of \(\Cal{T}\),
	and as \(\Cal{T}\) is length-preserving,
	we have that \(s(\rho_{\Rm{init}}\rho\rho_{\Rm{final}}) = s(\rho_{\Rm{init}}\rho^2\rho_{\Rm{final}}) = 0\). Consequently, \(s(\rho) = s(\rho^2) = 0\) and \(|u| = |v|\).
\end{proof}

\noindent
We now prove that the second condition of Lemma~\ref{lemma:bounded_equal_properties} is necessary.
\begin{lemma}
\label{lemma:bounded_imply_conjugate_cycle}
	Let \(\Cal{T} = (Q,\Sigma, Q_{i}, Q_{f}, \Delta)\) be a bounded NFT.
	Then for every run \(\rho \in \Delta^*\) from \(p \in Q\) to itself over some \((u,v)\),
	\(u\) is conjugate to \(v\) by \(s_p\).
\end{lemma}

\begin{proof}
	We prove  \cref{lemma:bounded_imply_conjugate_cycle} by contraposition.
	Let \(\rho \in \Delta^*\) be a run from \(p \in Q\) to itself over some \((u,v)\)
	such that \(u\) is not conjugate to \(v\) by \(s_p\).
	In particular, this means that \(u \neq \varepsilon\) or \(v \neq \varepsilon\), and by \cref{lemma:length_preserving_imply_u_equal_v}, \(|u| = |v| \neq 0\).
	Let \(k>0\) be an integer such that \(k|u|> |s_p|\).
	Let \(\rho_{\Rm{init}}\) be an initial run of \(\Cal{T}\) to \(p\) over some \((u',v')\) and let \(\rho_{\Rm{final}}\) be a final run of \(\Cal{T}\) from \(p\) over some \((u'',v'')\).
	By supposition, there exists \(i \leq |u|\) and \(j \leq |v|\) such that \(j-i \equiv s_p \bmod |u|\) and \(u_i \neq v_j\).
	For all \(m \in \mathbb{N}\) we can construct an accepting run
	\(\mu(m) = \rho_{\Rm{init}}\rho^{k+m}\rho_{\Rm{final}}\) over \(w(m)=u'u^{k+m}u''\) and \(z(m)=v'v^{k+m}v''\).
	By supposition, for each run \(\mu(m)\) there exists a set of positions 
	\[
		P = 
		\begin{cases}
			\{|v'| + (k + z -1) \cdot |v| + j \mid z \in [0,m]\} , &\text{ if }s_p \geq 0\\
			\{|u'| + (k + z -1) \cdot |u| + i \mid z \in [0,m]\} , &\text{ if }s_p < 0\\
			
		\end{cases}
	\]
	such that for all \(t \in P\) \(w(m)_{t} \neq z(m)_{t}\).
	Consequently, for all \(m \in \mathbb{N}\)
	we have a run \(\mu(m)\) of \(\Cal{T}\) over \((w(m),z(m))\) where \(d(w(m),z(m)) \geq m+1\).
	Therefore, \(\Cal{T}\) is not bounded.
	This establishes the contraposition and completes the proof.
\end{proof}


\subparagraph*{Sufficient conditions for boundedness.} 

\noindent
We now prove that transducers exceeding a quadratic deviation
already violate at least one of the two conditions,
which implies their sufficiency for boundedness.

\begin{lemma}\label{lemma:QuadBound_imply_not_properties}

Let \(\Cal{T} = (Q,\Sigma, Q_{i}, Q_{f}, \Delta)\) be an NFT and
\(b=\min(\smax(\Cal{T}) \cdot |Q|,|\Cal{T}|)\).

If \(\Cal{T}\) is not bounded by \(B=(b+\lmax(\Cal{T})+2)\cdot |Q|\). Then either:
\begin{enumerate}
\item \(\Cal{T}\) is not a length-preserving NFT, or\label{item:QuadBound_imply_not_properties_1}
\item \(\Cal{T}\) is length-preserving, but there exists \(\rho \in \Delta^*\),
a run from \(p \in Q\) to itself over some \((u,v) \in \Sigma^* \times \Sigma^*\)
such that \(u\) is not conjugate to \(v\) by \(s_p\).\label{item:QuadBound_imply_not_properties_2}
\end{enumerate}
\end{lemma}

\begin{proof}
	Let \(\Cal{T}\) be an NFT not bounded by \(B\).
	If it is also not length-preserving, then the lemma holds trivially.
	Suppose then that it is length-preserving.
	Let \(\rho = \delta_1\delta_2\cdots\delta_n \in \Delta^*\) be an accepting run of \(\Cal{T}\) over some \((u,v)\) such that \(d(u,v) \geq B\).
	We define \(m_1,\ldots,m_{d(u,v)}\) to be the strictly increasing sequence of positions such that \(u_{m_i} \neq v_{m_i}\) and \(q_i=\pi_1(\inn_\rho(m_i))\) the sequence of starting states of the transitions reading the mismatches.
	Since \(d(u,v) \geq B\), there exists a state \(p \in Q\) that appears at least \(b+\lmax(\Cal{T})+2\) times in the sequence \((q_i)_i\). We denote by \(o\) the smallest index of an occurrence of \(p\) and by \(r\) the \((\lmax(\Cal{T})+1)\)-th largest index of an occurrence of \(p\).
	Consequently, the factor \(\rho' = \delta_o\cdots\delta_{r}\) of \(\rho\) is a run from \(p\) to itself over some \((u',v')\) such that at least \(b +1\) mismatches are read in \(u'\).
	Since \(s_p\) is bounded by \(b\) by~\cref{lemma:bounded_imply_one_shift_by_state} and the subrun \(\rho'\) generates more than \(b\) mismatches, at least one of them is produced by \(\rho'\).
	As the initial shift is \(s_p\), there exist two integers \(i \leq |u'|\) and \(j \leq |v'|\) such that \(j-i \equiv s_p \bmod |u'|\) and \(u'_i \neq v'_j\).
	Consequently, \(u'\) is not conjugate to \(v'\) by \(s_p\).
\end{proof}

\subparagraph*{Complexity} 

\noindent
We now leverage the two conditions of~\cref{lemma:bounded_equal_properties} to obtain
the NL decision procedure for boundedness  that proves \cref{prop:BoundedNFTHigh}.
The NL algorithm is a combination of the next lemmas: 
it first checks whether the NFT is length-preserving by checking if:
\begin{enumerate}
\item All acyclic runs are length-preserving (\cref{lemma:complexity_u_equal_v_for_small_run})
\item All cyclic runs are length-preserving (\cref{lemma:complexity_u_equal_v_for_cycle})
\end{enumerate}
If both procedures accept, it then applies the algorithm from~\cref{lemma:complexity_conjugate_cycle} to test the second condition of~\cref{lemma:bounded_equal_properties}.

Before defining the algorithms, we establish two technical lemmas
that bound the maximal deviation inside runs and the length of runs
in terms of states and endpoint shifts.

\begin{lemma}\label{lemma:length-preserving_imply_smax_limited}
    Let \(\Cal{T} = (Q,\Sigma, Q_{i}, Q_{f}, \Delta)\) be a length-preserving transducer.
    Then for all runs \(\rho \in \Delta^*\), \(\smax(\rho) \leq \smax(\Cal{T}) \cdot |Q|\).
\end{lemma}

\begin{proof}
	Let \(\rho \in \Delta^*\) be a run of a length-preserving NFT \(\Cal{T}\). By~\cref{lemma:length_preserving_imply_u_equal_v}, cycles do not act on the shift of a run. Then there exists \(\rho' \in \Delta^*\) a run of \(\Cal{T}\) without cycle such that \(\smax(\rho) = \smax(\rho')\).
	Being without cycle, \(|\rho'| \leq |Q|\) and by definition \(\smax(\rho') \leq |\rho'| \cdot \smax(\Cal{T}) \leq \smax(\Cal{T}) \cdot |Q|\).
\end{proof}

\begin{lemma}\label{lemma:length-preserving-imply-correlation-production-run}
    Let \(\Cal{T} = (Q,\Sigma, Q_{i}, Q_{f}, \Delta)\) be a length-preserving transducer, and let \(\rho \in \Delta^*\) be a run of \(\Cal{T}\) from a state \(q \in Q\) to a state \(p \in Q\) over some \((u,v)\) without \((\varepsilon,\varepsilon)\)-cycles. Then
\[|\rho| \leq  |Q| \cdot |uv| = |Q| \cdot (s_q + 2 \cdot |u| - s_p) = |Q| \cdot (s_p + 2 \cdot |v| - s_q).\]

\end{lemma}

\begin{proof}
	Let \(\rho \in \Delta^*\) be a run of \(\Cal{T}\) from a state \(q \in Q\) to a state \(p \in Q\) over some \((u,v)\).
	Since there is no \((\varepsilon,\varepsilon)\)-cycles, each sequence of \(|Q|\) consecutive transitions either reads or writes at least one letter, hence \(|\rho| \leq |Q| \cdot |uv|\).
	Then by definition of the shift and the shift of a state, we have that \(|u|-|v|=s_p-s_q\).
	We obtain directly that \( |Q| \cdot |uv| = |Q| \cdot (s_q + 2 \cdot |u| - s_p) = |Q| \cdot (s_p + 2 \cdot |v| - s_q)\).
\end{proof}

\noindent
Before presenting the three NL procedures,
we briefly outline the overall strategy.
In each case, we show a small witness property:
whenever there exists a run violating the considered condition,
there also exists such a run of polynomially bounded size.
This allows us to nondeterministically guess a witness in logarithmic space,
and verify it in NL.

\begin{lemma}\label{lemma:complexity_u_equal_v_for_small_run}
	Given an NFT \(\Cal{T} = (Q,\Sigma, Q_{i}, Q_{f}, \Delta)\), it can be decided in NL
	whether every accepting runs \(\rho \in \Delta^*\) over some \((u,v)\)
	satisfies \(|\rho| \leq |Q|\) and \(|u| = |v|\).
\end{lemma}

\begin{proof}
	We construct a co-NL algorithm that accepts runs \(\rho \in \Delta^*\) over \((u,v) \in (\Sigma^* \times \Sigma^*)\) such that \(|\rho| \leq |Q|\) and \(|u| \neq |v|\).
	The algorithm works as follows.
	It guesses a run of \(\Cal{T}\) transition by transition and each time, updates two integers stored in binary:
	\begin{itemize}
		\item \(\ell\) stores the number of transitions that have already been through.
		\item \(s\) stores the shift between the two words already read and written.
	\end{itemize}
	If a final state is reached and \(s\) is not equal to \(0\) then the algorithm stops and accepts.
	If \(\ell\) become greater than \(|Q|\) then the algorithm stops and rejects.
	If neither of the above conditions are met, then the algorithm guesses the next transition, if it is not possible, then it stops and rejects.
	Denote that because we have been through at most \(|Q|\) transitions, \(s\) stays in the range \([ -\smax(\Cal{T}) \cdot |Q|,\smax(\Cal{T}) \cdot |Q|]\).
	Thus, this algorithm is NL.
\end{proof}

\begin{lemma}\label{lemma:complexity_u_equal_v_for_cycle}
    Given an NFT \(\Cal{T} = (Q,\Sigma, Q_{i}, Q_{f}, \Delta)\),
    it can be decided in NL whether every run \(\rho \in \Delta^*\)
    from some \(p \in Q\) to itself over some \((u,v)\) satisfies \(|u| \neq |v|\).
\end{lemma}

\begin{proof}
	First, note that if there exists a cyclic run of length greater than \(|Q|\) over some \((u,v)\)
	such that \(|u| \neq |v|\),
	decomposing it into simple cycles will yield at least one cyclic subrun
	of length smaller than or equal to \(|Q|\)
	over some \((u',v')\)
	such that \(|u'| \neq |v'|\).
	We construct a co-NL algorithm that accepts runs \(\rho \in \Delta^*\) with \(|\rho| \leq |Q|\) from \(p \in Q\) to itself, over some \((u,v)\) such that \(|u| \neq |v|\).
	The algorithm works as follows.
	We guess a state \(p\) and a run \(\rho\) of \(\Cal{T}\) starting in \(p\) transition by transition. At each step it updates three integers stored in binary:
	\begin{itemize}
		\item \(p\) stores the first state of \(\rho\).
		\item \(\ell\) stores the number of transitions that have already been processed.
		\item \(s\) stores the shift between the processed input and output.
	\end{itemize}
	If the state \(p\) is reached and \(s\) is not equal to \(0\) then the algorithm stops and accepts.
	If \(\ell\) become greater than \(|Q|\) then the algorithm stops and rejects.
	If neither of the above conditions are met, then the algorithm guesses the next transition if possible, and rejects otherwise.
	Note that since the algorithm processes at most \(|Q|\) transitions, \(s\) stays in the range \([- \smax(\Cal{T}) \cdot |Q|, + \smax(\Cal{T}) \cdot |Q|]\).
	Thus, this algorithm is NL.
\end{proof}

\begin{lemma}\label{lemma:complexity_conjugate_cycle}
    Given a length-preserving NFT \(\Cal{T} = (Q,\Sigma, Q_{i}, Q_{f}, \Delta)\),
    it can be decided in NL whether for all cyclic runs \(\rho \in \Delta^*\) over some \((u,v)\),
    \(u\) is conjugate to \(v\) by \(s_p\).
\end{lemma}

\begin{proof}
We first claim a small witness property and give an NL algorithm relying on it. 
The claim is proved afterward.
	\begin{claim}\label{claim-complexity_conjugate_cycle}
		If there exists a run \(\rho \in \Delta^*\) from \(p \in Q\) to itself over some \((u,v)\) such that \(u\) is not conjugate to \(v\) by \(s_p\) then there exists one of length inferior or equal to \(L = 2 \cdot |Q| + 2 \cdot \smax(\Cal{T}) \cdot |Q|^2\).
	\end{claim}
	From this claim, we construct a co-NL algorithm that accepts runs \(\rho \in \Delta^*\) of length \(|\rho| \leq L\) from \(p \in Q\) to itself, over some \((u,v)\) such that \(u\) is not conjugate to \(v\) by \(s_p\).
	The algorithm works as follows. We first guess a state \(p \in Q\), and we use an oriented graph accessibility algorithm to find an initial run leading to \(p\) and compute the value \(s_p\) along the way.
	Being the sum of the shift of each transition of this initial run, \(s_p\) is computable in NL.
	Next we guess two integers \(i\) and \(j\) and a run of \(\Cal{T}\) transition by transition. At each step we update five values:
	\begin{itemize}
		\item \(\ell\) stores the number of transitions already have been processed.
		\item \(\ell_r\) stores the length of the already processed input.
		\item \(\ell_w\) stores the length of the already processed output.
		\item \(u_i\) store the i-th letter of the word read.
		\item \(v_j\) store the j-th letter of the word written.
	\end{itemize}
	If \(\ell\) becomes greater than \(L\) before the state \(p\) is reached then we reject.
	If the algorithm reaches the state \(p\) with \(j-i \not\equiv s_p \bmod |\ell_r|\) then it rejects.
	Finally, if the state \(p\) is reached with \(j-i \equiv s_p \bmod |\ell_r|\) and \(u_i \neq v_j\) then it accepts.
	Note that since the algorithm processed at most \(L\) transitions, \(\ell_r\) and \(\ell_w\) both stays in the range \([- \lmax(\Cal{T}) \cdot L, \lmax(\Cal{T}) \cdot L]\).
	Thus, this algorithm is NL.

	\proofsubparagraph{Proof of \cref{claim-complexity_conjugate_cycle}.}
	Let \(\rho = \delta_1\delta_2\cdots\delta_n \in \Delta^*\) be a run from \(p \in Q\) to itself over some \((u,v)\) such that \(u\) is not conjugate to  \(v\) by \(s_p\), i.e. there exists \(i,j \in [1,|u|]\) such that \(j - i \equiv s_p \bmod |u|\), and \(u_i \neq v_j\).
	Remark that by sufficiently iterating \(\rho\), we can ensure that we are in one of the following cases:
	\begin{enumerate}
		\item \(s_p \geq 0\) and \(j > s_p \),
		\item \(s_p \leq 0\) and \(i > |s_p| \).
	\end{enumerate}
	Those two cases being symmetrical, we focus on the first one.
	We denote by \(r\) and \(w\) \(\inn(i)\) and \(\out(j)\) respectively.
	We next decompose \(\rho\) into five parts: three subruns, and two transitions:
	\begin{enumerate}
		\item \(\rho^{s} = \delta_1\cdots\delta_{\min(r,w)-1}\), a run from \(p\) to \(q'\) over some \((u',v')\),
		\item \(\delta^{o} = \delta_{\min(r,w)}\), a transition from \(p''\) to \(q''\) over some \((u'',v'')\),
		\item \(\rho^{m} = \delta_{\min(r,w)+1}\cdots\delta_{\max(r,w)-1}\), a run from \(p^{(3)}\) to \(q^{(3)}\) over some \((u^{(3)},v^{(3)})\),
		\item \(\delta^{r} = \delta_{\max(r,w)}\), a transition from \(p^{(4)}\) to \(q^{(4)}\) over some \((u^{(4)},v^{(4)})\),
		\item \(\rho^{f} = \delta_{\max(r,w)+1}\cdots\delta_n\), a run from \(p^{(5)}\) to \(p\) over some \((u^{(5)},v^{(5)})\).
	\end{enumerate}
	Informally,  \(\delta^o\) and  \(\delta^r\) are the transitions that read and write the positions that witness the non conjugacy, while \(\rho^{s}\), \(\rho^{m}\) and \(\rho^{f}\) decompose the remainder of \(\rho\) into respectively before, between, and after these transitions, as shown in \cref{fig:complexity_conjugated_cycle_fig}.
	Since \(\delta^o\) and \(\delta^r\) are not in neither \(\rho^{s}\) nor \(\rho^{f}\), if their length is greater than \(|Q|\) then there exists a factor that is a cyclic run.
	Since \(\Cal{T}\) is length-preserving, this factor can be removed without modifying the shift. Then we can assume \(\rho^{s}\) and \(\rho^{p}\) to be of length smaller than \(|Q|\).
	
	Concerning \(\rho^m\), again we distinguish three cases, either \(r = w\), \(r < w\), or \(r > w\).
	The first case means that \(\delta^1=\delta^2\) and \(\rho^m\) is inexistent.
	The two other cases are symmetrical, we focus on the first one, when \(r < w\).
	We denote by \(o \in Q\) the starting state of \(\rho^m\), \(q \in Q\) the final state of \(\rho^m\), and \((u',v')\) respectively the input and output words of \(\rho^m\).
	As \(r < w\), \cref{lemma:length-preserving_imply_smax_limited} implies that \(0 < s_{p^{(3)}} \leq \smax(\Cal{T}) \cdot |Q|\), \(0 < s_{q^{(3)}} \leq \smax(\Cal{T}) \cdot |Q|\), and \(|u^{(3)}| < s_{p^{(3)}}\).
	Therefore, using~\cref{lemma:length-preserving-imply-correlation-production-run}, \(|\delta^m| \leq |Q| \cdot |u^{(3)}v^{(3)}| = |Q| \cdot (s_{q^{(3)}} + 2 \cdot |u^{(3)}| - s_{p^{(3)}})\).
	Combining these equations, we obtain that \(|\delta^m| \leq 2 \cdot \smax(\Cal{T}) \cdot |Q|^2 - 2\).
	To conclude, if there exists a run \(\rho\) from \(p \in Q\) to itself over some \((u,v)\) such \(u\) is not conjugate to  \(v\) by \(s_p\), then there exists one of size smaller than or equal to:
	\[
		2 \cdot |Q| + 2 + 2 \cdot \smax(\Cal{T}) \cdot |Q|^2 - 2 = 2 \cdot |Q| + 2 \cdot \smax(\Cal{T}) \cdot |Q|^2 = L.\qedhere
	\]
\end{proof}

\begin{figure}
	\centering
	\resizebox{0.9\textwidth}{!}{

\begin{tikzpicture}

\draw[very thick] (0,0) rectangle +(13,1);
\draw[thick,fill=blue!20] (0,0) rectangle +(8,1);
\draw[thick,fill=orange!30] (8,0) rectangle +(1,1);
\draw[thick,fill=red!40] (9,0) rectangle +(2,1);
\draw[thick,fill=green!40] (11,0) rectangle +(1,1);
\draw[thick,fill=yellow!40] (12,0) rectangle +(1,1);

\path[draw,very thick,latex-latex,dashed] (13,0.5) to node[above]
     {\(s_p\)} (16,0.5);

\node(w) at (-2,0.5) {\(v\) (write)};

\node(rw) at (4,0.5){\textcolor{black}{\(v'\)}};
\node(drw) at (8.5,0.5){\textcolor{black}{\(v''\)}};
\node(r2w) at (10,0.5){\textcolor{black}{\(v^{(3)}\)}};
\node(drr) at (11.5,0.5){\textcolor{black}{\(v^{(4)}\)}};
\node(r3w) at (12.5,0.5){\textcolor{black}{\(v^{(5)}\)}};

\node (j) at (8.5,-0.5) {\textcolor{red}{\(j\)}};
\begin{scope}[xshift=3cm,yshift=2cm]
\draw[very thick] (0,0) rectangle +(13,1);
\draw[thick,fill=blue!20] (0,0) rectangle +(2,1);
\draw[thick,fill=orange!30] (2,0) rectangle +(1,1);
\draw[thick,fill=red!40] (3,0) rectangle +(2,1);
\draw[thick,fill=green!40] (5,0) rectangle +(1,1);
\draw[thick,fill=yellow!40] (6,0) rectangle +(7,1);

\path[draw,very thick,latex-latex,dashed] (-3,0.5) to node[above]
     {\(s_p\)} (0,0.5);
     
\node(w) at (-5,0.5) {\(u\) (read)};  
\node(rr) at (1,0.5){\textcolor{black}{\(u'\)}};
\node(dww) at (2.5,0.5){\textcolor{black}{\(u''\)}};
\node(r2r) at (4,0.5){\textcolor{black}{\(u^{(3)}\)}};
\node(dwr) at (5.5,0.5){\textcolor{black}{\(u^{(4)}\)}};
\node(r3r) at (9.5,0.5){\textcolor{black}{\(u^{(5)}\)}};

\node (i) at (5.5,1.5) {\textcolor{red}{\(i\)}};
\end{scope}

\path[draw,dashed,red] (i) to node[left] {\textcolor{black}{mismatch}} (j);

\end{tikzpicture}}
	\caption{Illustration of the decomposition of \(\rho\) into five parts. In this case, \(\min(r,w) = w\), and \(\max(r,w) = r\).}
	\label{fig:complexity_conjugated_cycle_fig}
\end{figure}

	
\subsection{ Threshold-bounded Deviation Problem} 
\label{ssub:k_bounded_nft}

Now we address the Threshold-bounded Deviation Problem,
and show that its complexity depends on whether the threshold \(k\) is part of the input or fixed.

\begin{proposition}\label{prop:KboundedNFTHigh}
	The Threshold-bounded Deviation Problem is:
	\begin{itemize}
		\item in co-NP when \(k\) is part of the input;
		\item in NL for every fixed \(k \in \mathbb{N}\).  
		\end{itemize}
\end{proposition}

Since a NFT \(\Cal{T} = (Q,\Sigma, Q_{i}, Q_{f}, \Delta)\) can be bounded only if it is length-preserving,
then \(\Cal{T}\) is \(k\)-bounded if and only if for all accepting runs \(\rho \in \Delta^*\) over some \((u,v)\) there is at most \(k\) mismatch between \(u\) and \(v\).
We call a run with \(k\) mismatch a \(k\)\emph{-mismatch witness} for \(\Cal{T}\).
We show that if a \(k\)-mismatch witness exists, there exists one of polynomial length on the size of \(\Cal{T}\) and consequently, we can construct an NP algorithm to find such witness.
More precisely, the algorithm first checks whether the NFT is bounded, which can be done in NL thanks to \cref{thm-BoundedNFT}.
Then if \(k\) is greater than the quadratic bound of \cref{lemma:QuadBound_imply_not_properties} the algorithm accepts, otherwise it applies the co-NP algorithm of the Lemma below.

\begin{lemma}\label{lemma:k-bounded-NP-algo}
    Given a bounded NFT \(\Cal{T} = (Q,\Sigma, Q_{i}, Q_{f}, \Delta)\), and an integer \(k <  B\) where \(B = (\min(\smax(\Cal{T}) \cdot |Q|,|\Cal{T}|)+\lmax(\Cal{T})+2)\cdot |Q|\), it can be decided in co-NP whether for all accepting runs \(\rho \in \Delta^*\) over some \((u,v)\), \(d(u,v) \leq k\).
\end{lemma}

\begin{proof}
	We first claim a small witness property and give a co-NP algorithm relying on it.
	The claim is proved afterward.

	\begin{claim}\label{claim:k-bounded-NP-algo}
		If there exists an accepting run \(\rho \in \Delta^*\) over some \((u,v)\) such that \(d(u,v) > k\) then there exists an accepting run \(\rho' \in \Delta^*\) over some \((u',v')\) such that \(|\rho'| \leq  L = 8 \cdot \smax(\Cal{T}) \cdot |Q|^3\) and \(d(u',v') > k\).
	\end{claim}

	From this claim, we construct a co-NP algorithm that accepts accepting runs \(\rho \in \Delta^*\) of length \(|\rho| \leq L\) over some \((u,v)\) such that \(d(u,v) > k\).
	The algorithm works as follows.
	It first guesses \(k+1\) positions, all these positions are stored in a set of pairs \(Z\) of the form \((i,a)\), where \(i\) is an integer storing the guessed position, and \(a\) stores an element of \(\Sigma \cup \{\varepsilon\}\).
	Initially, all pairs in \(Z\) are of the form \((i,\varepsilon)\).
	Next, the algorithm guesses a run \(\rho\) of \(\Cal{T}\), transition by transition, and each step it updates four variables:
	\begin{itemize}
		\item \(\ell\), an integer stored in binary, tracking the number of processed transitions.
		\item \(\ell_u\), an integer stored in binary, tracking the length of the processed input. 
		\item \(\ell_v\), an integer stored in binary, tracking the length of the processed output. 
		\item \(Z\), the set of pairs.
	\end{itemize} First, if the first state of the run \(\rho\) is not in \(Q_i\) then the algorithm stops and rejects.
	Next, at each step, the algorithm starts by comparing the input and output of the current transition denoted \((u,v)\) with \(Z\).
	If there exists \((i,a) \in Z\) such that \(i\) is in \((\ell_u,\ell_u + |u|]\), or in \((\ell_v,\ell_v + |v|]\) then the following updates are executed:
	\begin{itemize}
		\item If \(a = \varepsilon\), and \(i\) is in both \( (\ell_u,\ell_u + |u|]\) and \((\ell_v,\ell_v + |v|]\). Then, if \(u_{i-\ell_u} \neq v_{i-\ell_v}\), the algorithm removes \((i,\varepsilon)\) from the set, and if \(u_{i-\ell_u} = v_{i-\ell_v}\), it stops and rejects.
		\item Else if \(a = \varepsilon\), then the algorithm updates the value of \(a\) with \(u_{i-\ell_u}\) or \(v_{i-\ell_v}\) depending on if \(i\) was in \((\ell_u,\ell_u + |u|]\), or in \((\ell_v,\ell_v + |v|]\).
		\item Else if \(a \neq \varepsilon\) and \(a\) is not equal to \(u_{i-\ell_u}\) or \(v_{i-\ell_v}\), again depending on if \(i\) was in \((\ell_u,\ell_u + |u|]\), or in \((\ell_v,\ell_v + |v|]\) then the algorithm removes \((i,a)\) from the set.
		\item Else if \(a \neq \varepsilon\) and \(a\) is equal to \(u_{i-\ell_u}\) or \(v_{i-\ell_v}\), again depending on if \(i\) was in \((\ell_u,\ell_u + |u|]\), or in \((\ell_v,\ell_v + |v|]\) then the algorithm stops and rejects.
	\end{itemize}
	After these updates of \(Z\) the algorithm updates the three other variables \(\ell\), \(\ell_u\), and \(\ell_v\).
	If a final state is reached with \(\ell \leq L\) and \(|Z| = 0\) then the algorithm stops and accepts.
	If \(\ell\) becomes greater than \(L\), then the algorithm stops and rejects.
	If neither of this two conditions are met, then the algorithm guesses the next transition, if it is not possible, then it stops and rejects.

	\proofsubparagraph{Proof of~\cref{claim:k-bounded-NP-algo}}.
	Let \(\rho \in \Delta^*\) be an accepting run of \(\Cal{T}\) over some \((w,z)\) such that \(k < d(w,z)\), and \(|\rho| > L\).
	We prove that there is a valid sub-run \(\rho'\) of \(\rho\) that generates the same number of mismatches as \(\rho\).
	We assume that \(\rho\) is \((\varepsilon, \varepsilon)\)-cycle free, as such cycles cannot generate mismatches nor shift.
	Consequently, there exists a state \(p \in Q\) which is visited \(8 \cdot \smax(\Cal{T}) \cdot |Q|^2\), and there exists \(\mu =  \delta_1\delta_2\cdots\delta_n\) a factor of \(\rho\) from \(p\) to itself over some \((u,v)\) with \(|\mu| \geq 8 \cdot \smax(\Cal{T}) \cdot |Q|^2\).
	We distinguish two cases, when \(s_p \geq 0\) and when \(s_p \leq 0\).
	Those two cases being symmetrical, we focus on the first one.
	Using \cref{lemma:length-preserving-imply-correlation-production-run}, we get that \(|uv| \leq 8 \cdot \smax(\Cal{T}) \cdot |Q|\).
	Moreover, thanks to \cref{lemma:length-preserving_imply_smax_limited}, \(||u| - |v|| \leq \smax(\Cal{T}) \cdot |Q|\), therefore, we get that both \(|u| > \smax(\Cal{T}) \cdot |Q|\) and \(|v| > \smax(\Cal{T}) \cdot |Q|\).
	We decompose the run \(\mu\) into five part, three runs and two transitions, as shown in \cref{fig:k-bounded-NP-algo-witness-size-fig}:
	\begin{enumerate}
		\item \(\mu^s = \delta_1\cdots\delta_{\out(s_q)-1}\), a run from \(p\) to \(q'\) over some \((u',v')\),
		\item \(\delta^o = \delta_{\out(s_q)}\), a run from \(p''\) to \(q''\) over some \((u'',v'')\),
		\item \(\mu^m = \delta_{o+1}\cdots\delta_{m-1}\), a run from \(p^{(3)}\) to \(q^{(3)}\) over some \((u^{(3)},v^{(3)})\),
		\item \(\delta^r = \delta_{\inn(|u'| - s_q +1)}\), a run from \(p^{(4)}\) to \(q^{(4)}\) over some \((u^{(4)},v^{(4)})\),
		\item \(\mu^f = \delta_{\inn(|u'| - s_q +1)+1}\cdots\delta_n\), a run from \(p^{(5)}\) to \(p\) over some \((u^{(5)},v^{(5)})\).
	\end{enumerate}
	Informally, \(\delta^o\) marks the transition where \(s_p\), the initial shift \(\mu\) is caught up, and \(\delta^r\) marks the symmetrical transition from which the input read is not matched within \(\mu\). 
	Note that the output of \(\mu^s\) is equal to \(s_q\) which is smaller than \(\smax(\Cal{T}) \cdot |Q|\).
	Therefore, using~\cref{lemma:length-preserving-imply-correlation-production-run}, we can suppose that \(\mu^s\) is a \((\varepsilon,\varepsilon)\)-cycle free run of length at most \(|Q|(s_q + 2 s_q)  \leq 3\smax(\Cal{T}) \cdot |Q|^2\).
	A symmetric argument holds for \(\mu^f\) and gives the same bound on its length.
	Finally, the length of \(\mu^m\) can be bounded by \(|Q|\) thanks to~\cref{lemma:bounded_imply_conjugate_cycle} as any loop in \(\mu^m\) can be cut off without affecting the number of mismatches.
	
Altogether, we proved that \(\mu\) can be made to be of length smaller than 
\[
|\mu^s|+1+|\mu^m|+1+|\mu^f| \leq 6\smax(\Cal{T}) \cdot |Q|^2 + |Q| + 1 \leq 8\smax(\Cal{T}) \cdot |Q|^2.\qedhere
\]
\end{proof}

\begin{figure}
	\centering
	\resizebox{0.9\textwidth}{!}{

\begin{tikzpicture}

\draw[very thick] (0,0) rectangle +(13,1);
\draw[thick,fill=blue!20] (0,0) rectangle +(5,1);
\draw[thick,fill=orange!30] (5,0) rectangle +(1,1);
\draw[thick,fill=red!40] (6,0) rectangle +(4,1);
\draw[thick,fill=green!40] (10,0) rectangle +(1,1);
\draw[thick,fill=yellow!40] (11,0) rectangle +(2,1);

\path[draw,very thick,latex-latex,dashed] (13,0.5) to node[above]
     {\(s_p\)} (18.5,0.5);

 \node(w) at (-2,0.5) {\(v\) (write)};

 \node(m1w) at (2.5,0.5){\textcolor{black}{\(v'\)}};
 \node(d0w) at (5.5,0.5){\textcolor{black}{\(v''\)}};
 \node(m2w) at (8,0.5){\textcolor{black}{\(v^{(3)}\)}};
 \node(dmw) at (10.5,0.5){\textcolor{black}{\(v^{(4)}\)}};
 \node(r3w) at (12,0.5){\textcolor{black}{\(v^{(5)}\)}};

 \begin{scope}[xshift = +5.5cm, yshift=+2cm]

\draw[very thick] (0,0) rectangle +(13,1);
\draw[thick,fill=blue!20] (0,0) rectangle +(2,1);
\draw[thick,fill=orange!30] (2,0) rectangle +(1,1);
\draw[thick,fill=red!40] (3,0) rectangle +(4,1);
\draw[thick,fill=green!40] (7,0) rectangle +(1,1);
\draw[thick,fill=yellow!40] (8,0) rectangle +(5,1);
\node(w) at (-7,0.5) {\(u\) (read)};

\path[draw,very thick,latex-latex,dashed] (-5.5,0.5) to node[above]
     {\(s_p\)} (0,0.5);
 \node(m1r) at (1,0.5){\textcolor{black}{\(u'\)}};
 \node(d0r) at (2.5,0.5){\textcolor{black}{\(u''\)}};
 \node(m2r) at (5,0.5){\textcolor{black}{\(u^{(3)}\)}};
 \node(dmr) at (7.5,0.5){\textcolor{black}{\(u^{(4)}\)}};
 \node(r3r) at (10.5,0.5){\textcolor{black}{\(u^{(5)}\)}};

 \end{scope}

\path[draw,thick,dashed] (5.5,2) -- (5.5,0);
\path[draw,thick,dashed] (13,1) -- (13,3);
\end{tikzpicture}}
	\caption{Illustration of the decomposition of \(\mu\) into five parts.}
	\label{fig:k-bounded-NP-algo-witness-size-fig}
\end{figure}



\section{Conclusion} 
\label{sec:conclusion}

In this work, we introduced some new lower-bounds on the complexity of three NFTs comparison problems for pairs of transducers.
We also strengthened the already studied upper-bounds for those problems in \cite{AiswaryaMS24LongVersion}.
A natural direction for pursuing this work would be to extend these positive results to other metrics over pair of words, such as the Levenshtein edit distance~\cite{Levenshtein66} or the 
Dynamic Time Warping distance (or DTW)~\cite{Vintsyuk,saoke78}
which is notably used in speech recognition or DNA sequencing for the comparison of gene expression~\cite{SVBKP13}.

In another direction, the NL-complexity that we achieved for the Bounded Comparison Problem and more importantly the fixed Threshold Comparison Problem could be leveraged for practical use.
In particular, one can think of applying these to the fixed approximated Model-Checking,
which asks whether a specification can be satisfied by a class of machine up to tolerating a fixed maximal number of errors.
In our setting, the specification is given by the first transducer which can be highly ambiguous, while the model is given by the second transducer which would be a deterministic one. Deciding the fixed threshold comparison amounts then to decide the fixed approximated Model-Checking.
In a similar but more ambitious vein, we hope to extend this approach to approximated synthesis, by adapting our approach to directly generate, from the specification given as an ambiguous transducer, a deterministic transducer whose deviation from the specification is bounded.


\newpage
\bibliography{biblio}

\end{document}